\documentclass[aip,jmp,preprint,nofootinbib]{revtex4-1}

\draft 

\usepackage[table]{xcolor}
\usepackage{hhline}
\usepackage{hyperref}
\usepackage{amsthm,amsfonts,amsmath, amscd}
\usepackage{amssymb,amsmath, dsfont}
\usepackage{ytableau}
\usepackage{physics}
\usepackage{braket}
\usepackage{dsfont}
\usepackage{tikz}
\usepackage[toc]{appendix}

\usepackage[T1]{fontenc}

\usepackage{array}
\usepackage{boldline}
\usepackage{enumitem}

\newcolumntype{P}[1]{>{\centering\arraybackslash}p{#1}}
\newcolumntype{M}[1]{>{\centering\arraybackslash}m{#1}}
\newcolumntype{a}[1]{!{\vrule width #1}c|}
\newcolumntype{b}[1]{!{\vrule width #1}c}

\newsavebox{\tempbox}

\newtheorem{theorem}{\bf Theorem}[]
\newtheorem{corollary}[theorem]{\bf Corollary}
\newtheorem{proposition}[theorem]{\bf Proposition}
\newtheorem{lemma}[theorem]{\bf Lemma}

\newtheorem{definition}[theorem]{\bf Definition}
\newtheorem{example}[theorem]{\bf Example}

\theoremstyle{remark}
\newtheorem{remark}[theorem]{\bf Remark}
\newtheorem*{remarks}{\bf Remarks}

\DeclareMathOperator{\Vol}{Vol}

\DeclareMathOperator{\SYM}{SYM}
\DeclareMathOperator{\sgn}{sgn}

\DeclareMathOperator{\ad}{ad}
\DeclareMathOperator{\Lie}{Lie}
\DeclareMathOperator{\diag}{diag}
\DeclareMathOperator{\SU}{SU}

\DeclareMathOperator{\GL}{GL}

\DeclareMathOperator{\SL}{SL}
\DeclareMathOperator{\spl}{\mathfrak{sl}}

\newcommand{\U}{\mathcal{U}}

\newcommand{\LME}{\text{\normalfont LME}}

\newcommand{\be}{\begin{equation}}
\newcommand{\ee}{\end{equation}}
\newcommand{\ba}{\begin{eqnarray}}
\newcommand{\ea}{\end{eqnarray}}
\newcommand{\bes}{\begin{equation}}
\newcommand{\ees}{\end{equation}}
\newcommand{\bas}{\begin{eqnarray*}}
\newcommand{\eas}{\end{eqnarray*}}

\newcommand{\beq}{\begin{equation}}
\newcommand{\eeq}{\end{equation}}

\definecolor{slidegreen}{rgb}{0,.5,0}
\definecolor{slidered}{rgb}{1,0,0}
\definecolor{slideblue}{rgb}{0.1,0.1,0.8}
\definecolor{slideorange}{rgb}{1,0.5,0}

\def\red#1{\textcolor{slidered}{#1}}
\def\blue#1{\textcolor{slideblue}{#1}}

\begin{document}


\title{Generalized Pauli constraints in large systems: the Pauli principle dominates}



\author{Robin Reuvers}
\email[Email: ]{rr@math.ku.dk}
\affiliation{DAMTP, Centre for Mathematical Sciences, University of Cambridge, Wilberforce Road, Cambridge CB3 0WA, United Kingdom}
\affiliation{QMATH, Department of Mathematical Sciences, University of Copenhagen, Universitetsparken 5, DK-2100 Copenhagen \O, Denmark}

\begin{abstract}
Lately, there has been a renewed interest in fermionic 1-body reduced density matrices and their restrictions \textit{beyond} the Pauli principle. These restrictions are usually quantified using the polytope of allowed, ordered eigenvalues of such matrices.
Here, we prove this polytope's volume rapidly approaches the volume predicted by the Pauli principle as the dimension of the 1-body space grows, and that additional corrections, caused by generalized Pauli constraints, are of much lower order unless the number of fermions is small. Indeed, we argue the generalized constraints are most restrictive in (effective) few-fermion settings with low Hilbert space dimension.
\end{abstract}


\maketitle 
\section{Introduction}

Fermionic quantum states are antisymmetric: their $N$-body space is a wedge product $\wedge^N\mathcal{H}$ of $N$ copies of the 1-body space $\mathcal{H}$. In particular, this implies the Pauli principle.

Of course, antisymmetry is a more restrictive property, and it is a long-standing problem to find out just how restrictive it is [\onlinecite{coleman1963structure},\onlinecite{coleman2000reduced}]. For example, it is unknown what \textit{k-particle reduced density matrices}
\begin{equation}
\label{rdmdef}
\gamma^\Psi_k:=\binom{N}{k}\Tr_{k+1\dots N}[\ketbra{\Psi}]
\end{equation}
can arise from pure states $\ket{\Psi}\in\wedge^N\mathcal{H}\subset\otimes^N\mathcal{H}$. This is particularly relevant for $k=2$, since such knowledge would provide significant computational advantages.

In this paper, we focus on the simpler case $k=1$. The set of interest is
\begin{equation}
    \Big\{\gamma^\Psi_1\left|\ \ket{\Psi}\in\wedge^N\mathcal{H},\right.\|\ket{\Psi}\|=1\Big\}.
\end{equation}
Each $\gamma^\Psi_1$ is diagonalizable: it has eigenvalues and eigenvectors, but the latter can easily be changed with a unitary transformation. Indeed the set is closed under such transformations: it is entirely defined by the allowed eigenvalues of $\gamma^\Psi_1$.

For $\mathcal{H}=\mathbb{C}^d$ with $N\leq d$, this information amounts to
\begin{equation}
\label{fermiset}
F_{d,N}:=\left\{(\lambda_1,\dots,\lambda_d)\in\mathbb{R}^d\left|\  \text{$\lambda_1\geq\dots\geq\lambda_d$ eigenvalues of $\gamma^\Psi_1$ for $\ket\Psi\in\wedge^N\mathbb{C}^d$},\right.\|\ket{\Psi}\|=1\right\}.
\end{equation}
This is a convex polytope in $\mathbb{R}^d$ [\onlinecite{PauliRevisited}, \onlinecite{GS}, \onlinecite{NM}], and it can be determined numerically for small $N$ and $d$ [\onlinecite{PauliRevisited}].
Less is known about higher $N$ and $d$, and that is the focus of this paper. We are motivated by the ongoing attempts to use knowledge about $F_{d,N}$ in physics and chemistry [\onlinecite{benavides2013quasipinning},\onlinecite{benavides2018static},\onlinecite{klyachko2009},\onlinecite{schilling2015hubbard},\onlinecite{schilling2015quasipinning},\onlinecite{schilling2018generalized},\onlinecite{schilling2019implications}].

To start the investigation, let us check how $F_{d,N}$ relates to an important physical fact: the Pauli principle. The Pauli principle says that the expectation value of any particle number operator $n_i:=a^\dagger_ia_i$ in a normalized fermionic state $\ket\Psi$ is bounded by 1,
\begin{equation}
\braket{n_i}_{\Psi}=\bra{\Psi}a^\dagger_ia_i\ket{\Psi}=\bra{\Psi}\mathds{1}-a_ia^\dagger_i\ket{\Psi}=1-\|a^\dagger_i\ket\Psi\|^2\leq1,
\end{equation}
but this is equivalent to saying that the eigenvalues of $\gamma^{\Psi}_1$ are all bounded by $1$. After all, an annihilation operator $a_i$ acts as $\sqrt{N}(\bra{u_i}\otimes\mathds{1})$ on $N$-fermion states like $\ket\Psi$, for some 1-particle state $\ket{u_i}$, and
\begin{equation}
\label{enumber}
\bra{u_i}\gamma^\Psi_1\ket{u_i}=N\Tr_{1\dots N}[(\ketbra{u_i}\otimes\mathds{1})\ketbra{\Psi}]=\bra{\Psi}a^\dagger_ia_i\ket{\Psi}=\braket{n_i}\leq 1.
\end{equation}

Hence we know $\lambda_1\leq 1$ for points in $F_{d,N}$. Since $F_{d,N}$ is a convex polytope [\onlinecite{PauliRevisited}, \onlinecite{GS}, \onlinecite{NM}], it is completely defined by inequalities involving the $\lambda_i$. These are known are as \textit{generalized Pauli constraints}  [\onlinecite{PauliRevisited},\onlinecite{schilling2013pinning}], and have the general shape
\begin{equation}
    c_1 \lambda_1 +\dots+ c_d\lambda_d\leq b
\end{equation}
for $c_i,b\in\mathbb{R}$.
Are these as valuable as the Pauli inequality $\lambda_1\leq1$?

To investigate this from a purely mathematical viewpoint, define
\begin{equation}
P_{d,N}:=\Big\{(\lambda_1,\dots,\lambda_d)\in\mathbb{R}^d\left|\ \text{$1\geq\lambda_1\geq\dots\geq\lambda_d\geq0$ and $\lambda_1+\dots+\lambda_d=N$}\right.\Big\}.
\end{equation}
This is the crudest approximation to $F_{d,N}$ we can make, and it uses only the Pauli inequality and the normalization condition $\lambda_1+\dots+\lambda_d=N$. Does this give a good approximation to $F_{d,N}$? For low $N$ and $d$, certainly not.

\begin{example}
The $N=2$ case has been understood since 1961 [\onlinecite{yang1962},\onlinecite{youla1961}], the $N=3$, $d=6$ case since 1972 [\onlinecite{borland1972conditions}]. The relevant sets are
\[
\begin{aligned}
    F_{d,2}&=\Big\{(\lambda_1,\dots,\lambda_d)\left|\text{\normalfont\ $1\geq\lambda_1\geq\dots\geq\lambda_d\geq0$,\ \ $\sum_i\lambda_i=2$,\ \ $\lambda_{2i-1}=\lambda_{2i}$,\ \  $\lambda_d=0$ if $d$ is odd}\right.\Big\}\\
    F_{6,3}&=\Big\{(\lambda_1,\dots,\lambda_6)\left|\text{\normalfont\ $1\geq\lambda_1\geq\dots\geq\lambda_6\geq0$,\ \ $\sum_i\lambda_i=3$,\ \ $\lambda_i=1-\lambda_{7-i}$,\ \  $\lambda_4\leq\lambda_5+\lambda_6$}\right.\Big\}.
\end{aligned}
\]
In 2008, an algorithm was devised to calculate general $F_{d,N}$ [\onlinecite{PauliRevisited}]. The resulting polytopes for low $N$ and $d$ do not resemble $P_{d,N}$, but more so than in the instances above.  
\end{example}

Clearly the difference between $F_{d,N}$ and $P_{d,N}$ is huge in these cases. Does this remain true when $d$ increases, or when $N$ and $d$ both increase? One way to measure this is by comparing volumes. Although volume does not carry any physical information, it is a useful way to investigate the two polytopes. Indeed, in line with what is suggested by the explicit results for small $N$ and $d$ [\onlinecite{PauliRevisited}], we will show that $F_{d,N}$ and $P_{d,N}$ quickly have similar volume as $d$ increases, and explain why the polytopes are mostly alike.

The paper is divided into three parts. We discuss theorems about volume in Section \ref{volumetheorems}, important insights from the proof in Section \ref{discussion}, and the proof itself in Section \ref{proofs}.

\section{Theorems about volume}
\label{volumetheorems}
\subsection{Comparing the volumes of $F_{d,N}$ and $P_{d,N}$}
\label{iiA}
Recall that we want to compare
\begin{equation}
    \begin{aligned}
    F_{d,N}&=\left\{(\lambda_1,\dots,\lambda_d)\in\mathbb{R}^d\left|\  \text{$\lambda_1\geq\dots\geq\lambda_d$ eigenvalues of $\gamma^\Psi_1$ for $\ket\Psi\in\wedge^N\mathbb{C}^d$},\right.\|\ket{\Psi}\|=1\right\}\\
        P_{d,N}&=\Big\{(\lambda_1,\dots,\lambda_d)\in\mathbb{R}^d\left|\ \text{$1\geq\lambda_1\geq\dots\geq\lambda_d\geq0$ and $\lambda_1+\dots+\lambda_d=N$}\right.\Big\}.
    \end{aligned}
\end{equation}
Note that $\Vol^{d-1}(F_{d,N})=\Vol^{d-1}(F_{d,d-N})$ by particle-hole duality, and similar for $P_{d,N}$.

Our first theorem gives the volume's limiting behaviour. 
\begin{theorem}[Limit behaviour]
\label{maintheorem}
Let $N\geq 8$ be fixed. Then,
\[
    \lim_{d\to\infty}\frac{\Vol^{d-1}\big(F_{d,N}\big)}{\Vol^{d-1}\big(P_{d,N}\big)}=1.
\]
Alternatively, for a fixed filling ratio $r\in(0,1)$,
\[
    \lim_{d\to\infty}\frac{\Vol^{d-1}\big(F_{d,\lfloor rd\rfloor}\big)}{\Vol^{d-1}\big(P_{d,\lfloor rd\rfloor}\big)}=1.
\]
\end{theorem}
This theorem is a corollary of the following estimates, which are proved in Section \ref{proofs}.
\begin{theorem}[Quantitative estimate]
\label{fullvsPauli}
Let $8\leq N\leq d/2$ be fixed. Then, if $d$ is large enough to guarantee $d(\frac{N-1}{N})^{d-1}\leq1$,
\[
1\geq\frac{\Vol^{d-1}\big(F_{d,N}\big)}{\Vol^{d-1}\big(P_{d,N}\big)}\geq 1-\frac{d^N}{1-d(\frac{N-1}{N})^{d-1}}\left(\frac{\min\big[\frac{1}{2}(N+7),\sqrt{32N}\big]}{N}\right)^{d-1}.
\]
Also, for integers $d$ and $N=rd\geq20$ for some $r\in(0,1/2)$,
\[
1\geq \frac{\Vol^{d-1}\big(F_{d,rd}\big)}{\Vol^{d-1}\big(P_{d,rd}\big)}\geq 1-\frac{1}{r^{r+1/2}(1-r)^{3/2-r}}\left(\frac{8}{r^{r+1/2}(1-r)^{1-r}}\frac{1}{\sqrt{d}}\right)^{d-1}.
\]
\end{theorem}

\begin{remarks}
\begin{enumerate}
    \item Volume is used here as a way to compare $F_{d,N}$ and $P_{d,N}$---it does not carry any physical information. We do argue that insights from the proof allow us to draw some conclusions. These are discussed in Section \ref{discussion}.
    \item Although these estimates show that convergence occurs rapidly, we can obtain better estimates for low $N$ and $d$. Remark \ref{betterestimate} discusses this; Figure \ref{contourplot} illustrates the result.
    \item The ratios above concern the effect of the generalized Pauli constraints in excess of the Pauli principle. It is useful to compare this to the effect that the Pauli principle itself has on the bosonic analogue of $F_{d,N}$. We discuss this in the next subsection.
\end{enumerate}
\end{remarks}

\begin{figure}
        \centering
        \includegraphics[scale=0.6]{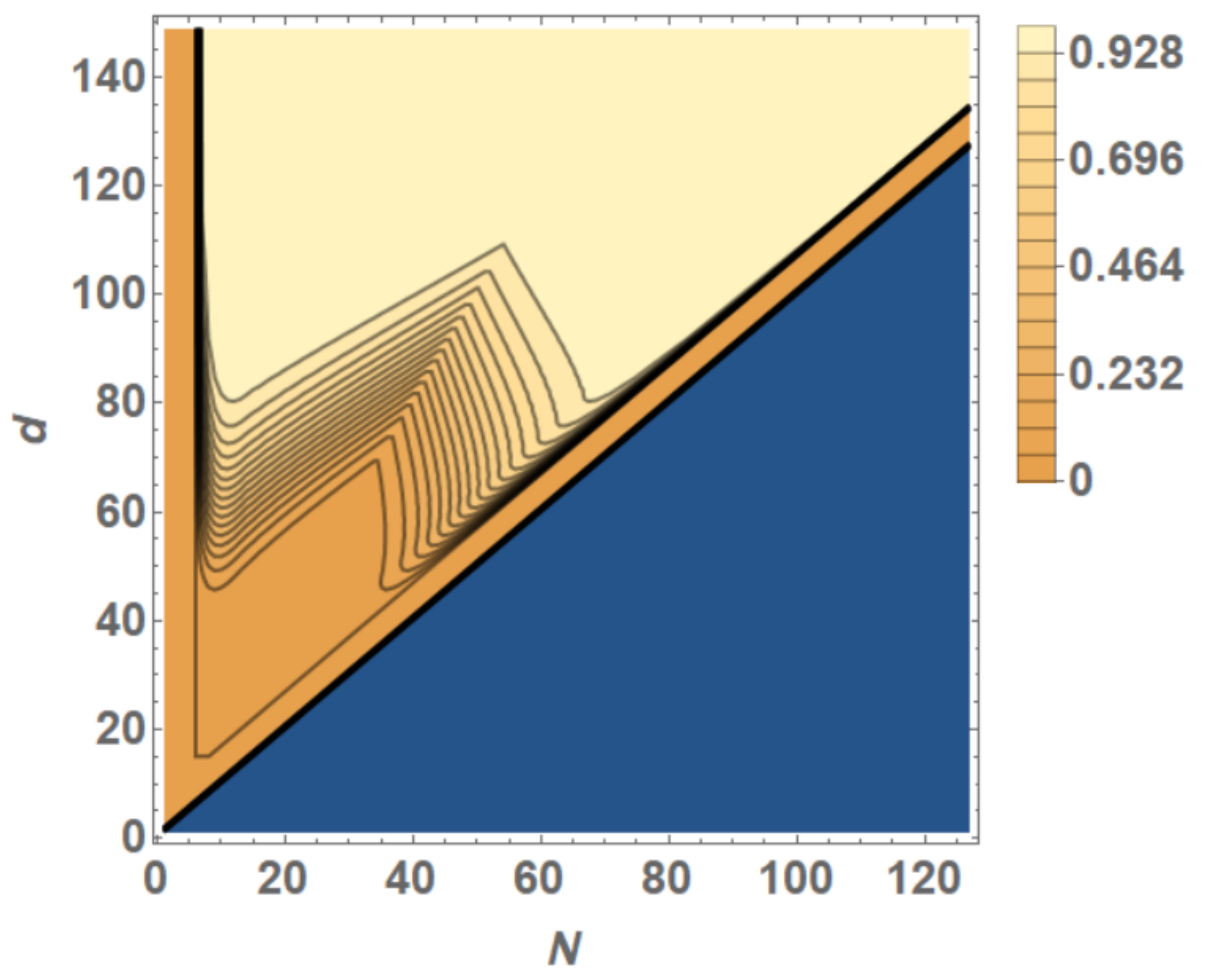}
        \caption{A contour plot demonstrating the lower bound to $\Vol^{d-1}(F_{d,N})/\Vol^{d-1}(P_{d,N})$ obtained in Remark \ref{betterestimate}. The blue part corresponds to $N>d$, which is not allowed. The convergence happens as orange turns to yellow, and it occurs extremely rapidly if $N\geq80$. Numerical simulations inspired by [\onlinecite{walteralgorithm0}] (now see  [\onlinecite{walteralgorithm}]) suggest convergence should actually happen more quickly in the region $8\leq N\leq 80$, so that the yellow region extends a long way towards the contour that forms a triangle in the orange region, but this cannot be demonstrated with our method. Similarly, we have no bound for $N,d-N\leq 8$, but numerics suggests rapid convergence in $d$ also for $4\leq N,d-N\leq 8$.}
    \label{contourplot}
    \end{figure}

\subsection{Comparing with the effect of the Pauli principle}
\label{comparison}
Define the bosonic polytope
\begin{equation}
\label{Boseset}
B_{d,N}:=\left\{(\lambda_1,\dots,\lambda_d)\in\mathbb{R}^d\left|\ \text{$\lambda_1\geq\dots\geq\lambda_d\geq0$ and $\lambda_1+\dots+\lambda_d=N$}\right.\right\},
\end{equation}
which is $P_{d,N}$ without the Pauli condition.
It is well known this set is physically correct for $N\geq 2$: it is equal to 
\begin{equation}
    \left\{(\lambda_1,\dots,\lambda_d)\in\mathbb{R}^d\left|\  \text{$\lambda_1\geq\dots\geq\lambda_d$ eigenvalues of $\gamma^\Psi_1$ for $\ket\Psi\in\otimes_{\SYM}^N\mathbb{C}^d$},\right.\|\ket{\Psi}\|=1\right\}.
\end{equation}
Indeed, the discrepancy between the `naive' $P_{d,N}$ and the correct, more complicated $F_{d,N}$ is a purely fermionic phenomenon. Nevertheless, it is useful to compare $B_{d,N}$ and $P_{d,N}$, since the Pauli principle cuts $B_{d,N}$ down to $P_{d,N}$, after which the generalized constraints cut $P_{d,N}$ down to $F_{d,N}$. It seems reasonable to compare the volumes lost in these two steps, as it suggests something about the impact of the generalized constraints compared to that of the Pauli principle. Let us stress again that this is the main motivation behind this work: volume itself is not important, but it is used here as a tool to investigate the structure of these polytopes.

To make a comparison, we first need information about the difference between $P_{d,N}$ and $B_{d,N}$. This is proved at the end of Section \ref{volumeestimates}.
\begin{proposition}[Volume loss due to Pauli]
\label{Pauliloss}
For $1\leq N\leq d$,
\[
1-d\left(\frac{N-1}{N}\right)^{d-1}\leq\frac{\Vol^{d-1}(P_{d,N})}{\Vol^{d-1}(B_{d,N})}\leq 1-\left(\frac{N-1}{N}\right)^{d-1}.
\]
\end{proposition}
This immediately implies two things. First, that for fixed $N$ and large $d$, the effect of the Pauli principle on volume is negligible, and second, that for a fixed ratio $r=N/d$, the Pauli principle has a non-negligible effect on volume. As we saw in Section \ref{iiA}, the generalized constraints have a negligible effect in all cases. Using Theorem \ref{fullvsPauli} and Proposition \ref{Pauliloss}, a quantitative comparison can be made. To do this, note
\begin{equation}
    \frac{\Vol^{d-1}\big(P_{d,N}\backslash F_{d,N}\big)}{\Vol^{d-1}\big(B_{d,N}\backslash P_{d,N}\big)}=\frac{1-\frac{\Vol^{d-1}\big(F_{d,N}\big)}{\Vol^{d-1}\big(P_{d,N}\big)}}{\frac{\Vol^{d-1}(B_{d,N})}{\Vol^{d-1}(P_{d,N})}-1},
\end{equation}
so that we obtain expressions like the ones in Theorem \ref{fullvsPauli}. Qualitatively nothing changes, except that $N^{d-1}$ gets replaced by
$(N-1)^{d-1}$ in the denominator of the first estimate. This says that the volume effect of the generalized Pauli constraints is much smaller than that of of the Pauli principle. The qualitative conclusions are listed below.

\begin{corollary}[Comparing to Pauli]
Let $N\geq 10$ be fixed. Then,
\[
    \lim_{d\to\infty}\frac{\Vol^{d-1}\big(P_{d,N}\big)}{\Vol^{d-1}\big(B_{d,N}\big)}=1\hspace{2cm} \lim_{d\to\infty}\frac{\Vol^{d-1}\big(P_{d,N}\backslash F_{d,N}\big)}{\Vol^{d-1}\big(B_{d,N}\backslash P_{d,N}\big)}=0.
\]
Also, for a fixed filling ratio $r\in(0,1)$,
\[
    \limsup_{d\to\infty}\frac{\Vol^{d-1}\big(P_{d,\lfloor rd\rfloor}\big)}{\Vol^{d-1}\big(B_{d,\lfloor rd\rfloor}\big)}\leq1-e^{-1/r}\hspace{1cm}\lim_{d\to\infty}\frac{\Vol^{d-1}\big(P_{d,\lfloor rd\rfloor}\backslash F_{d,\lfloor rd\rfloor}\big)}{\Vol^{d-1}\big(B_{d,\lfloor rd\rfloor}\backslash P_{d,rd}\big)}=0.
\]
\end{corollary}

\section{Insights from the proof}
\label{discussion}

\subsection{Proof strategy}
\label{proofstrategy}
\begin{enumerate}
    \item $P_{d,N}$ is a polytope. We first determine which of its extreme points lie in $F_{d,N}$. As discussed in Section \ref{extremepointssection}, this turns out to be the vast majority as $N$ and $d$ increase. However, these do not yet capture a volume.
    \item To deal with this, replace extreme points outside $F_{d,N}$ by one or more intermediate points that do lie in $F_{d,N}$. This captures part of the volume of $P_{d,N}$ by convexity, and this volume must be contained in $F_{d,N}$. In particular, we verify in Section \ref{AcontainedinFsubsection} that for integers $1\leq m\leq N-7$ and $t=\frac{N-m+1}{N-m+9}$,
    \begin{equation}
A_{d,N,m,t}:=\Big\{(\lambda_1,\dots,\lambda_d)\in P_{d,N}\left|\ \lambda_m\leq t\right.\Big\}\subset F_{d,N}.
\end{equation}
See Figure \ref{triangle} for an illustration. 
    \begin{figure}
        \centering
        \includegraphics[scale=0.5]{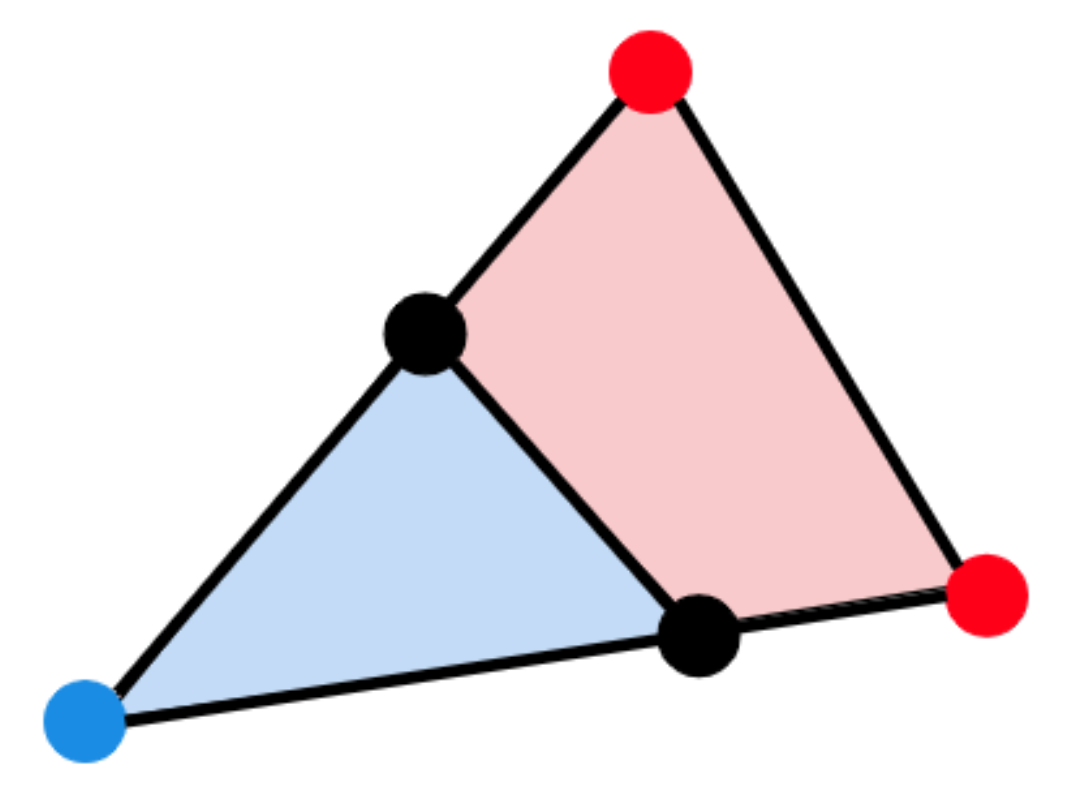}
        \caption{Suppose the triangle is $P_{d,N}$. It has three extreme points. Imagine that one (blue) is in $F_{d,N}$, whereas two (red) are not, and that we can verify that the black points are in $F_{d,N}$. This means $F_{d,N}$ contains the blue set, say $A_{d,N,m,t}$, and hence $\Vol^{d-1}(A_{d,N,m,t})\leq \Vol^{d-1}(F_{d,N})\leq \Vol^{d-1}(P_{d,N})$.}
            \label{triangle}
    \end{figure}

\item The above implies 
\begin{equation}
\frac{\Vol^{d-1}(F_{d,N})}{\Vol^{d-1}(P_{d,N})}\geq\frac{\Vol^{d-1}(A_{d,N,m,t})}{\Vol^{d-1}(P_{d,N})}.
\end{equation}
We estimate these volumes in Section \ref{volumeestimates} and prove Theorems \ref{maintheorem} and \ref{fullvsPauli}. Note that $m$ is a variable that can be optimized. 
\end{enumerate}
\begin{remark}
\label{betterestimate}
The volumes of $P_{d,N}$ and $A_{d,N,m,t}$ can also be calculated explicitly, see Proposition \ref{paulipolytope} and Appendix \ref{VolA} respectively. This is sharpest estimate our method can give, and it demonstrates how quickly $\Vol^{d-1}(F_{d,N})/\Vol^{d-1}(P_{d,N})$ converges to 1 already for low $N$ and $d$ (see Figure \ref{contourplot}).
\end{remark}

\subsection{Extreme points of $P_{d,N}$}
\label{extremepointssection}
We now discuss which extreme points of $P_{d,N}$ are also in $F_{d,N}$. This provides important clues as to why and when these two polytopes resemble each other. We start by indexing the extreme points of $P_{d,N}$.

\begin{proposition}[Properties of $P_{d,N}$]
\label{paulipolytope}
The extreme points of $P_{d,N}$ are the Slater point
\begin{equation}
\label{Slaterpoint}
\big(\underbrace{1,\dots,1}_N,\underbrace{0,\dots,0}_{d-N}\big)\in\mathbb{R}^d
\end{equation}
and the $N(d-N)$ distinct points
\begin{equation}
\label{extremefermions}
\big(\underbrace{\vphantom{\frac{N-i}{d-i-j}}1,\dots,1}_{i},\underbrace{\frac{N-i}{d-i-j},\dots,\frac{N-i}{d-i-j}}_{d-i-j},\underbrace{\vphantom{\frac{N-i}{d-i-j}}0,\dots,0}_{j}\big)\in\mathbb{R}^d \hspace{0.7cm} (0\leq i\leq N-1,\ 0\leq j\leq d-N-1).
\end{equation}
The polytope has $(d-1)$-dimensional volume
\begin{equation}
\Vol^{d-1}(P_{d,N})=\frac{1}{d!}\frac{\sqrt{d}}{(d-1)!}\sum^{N-1}_{k=0}(-1)^k\binom{d}{k}(N-k)^{d-1}.
\end{equation}
\end{proposition}

This is proved at the end of Section \ref{geometry}. For now, note that the extreme points of $P_{d,N}$ are completely defined by the fact that they have $i$ entries that are 1, and $j$ entries that are 0. As we discuss in Section \ref{fermionicstates}, the ones correspond to a Slater determinant that can be split off from the remainder of the state; the zeros can be ignored as unoccupied dimensions. This gives the following observation. 

\begin{proposition}
\label{extremepointtoLME}
For $0\leq i\leq N-1,\ 0\leq j\leq d-N-1$,
\[
\begin{aligned}
\big(\underbrace{\vphantom{\frac{N-i}{d-i-j}}1,\dots,1}_{i},\underbrace{\frac{N-i}{d-i-j},\dots,\frac{N-i}{d-i-j}}_{d-i-j},&\underbrace{\vphantom{\frac{N-i}{d-i-j}}0,\dots,0}_{j}\big)\in F_{d,N}\\
&\iff \Big(\underbrace{\frac{N-i}{d-i-j},\dots,\frac{N-i}{d-i-j}}_{d-i-j}\Big)\in F_{d-i-j,N-i}.
\end{aligned}
\]
\end{proposition}

States with this latter eigenvalue structure have been studied before. 

\enlargethispage{\baselineskip}
\begin{definition}[Completely entangled or fermionic LME states]
\label{LMEdef}
A normalized state $\ket\Psi$ in $\wedge^N\mathbb{C}^d$ is Locally Maximally Entangled (LME) [\onlinecite{Raamsdonk2},\onlinecite{Raamsdonk1}], alternatively, completely entangled [\onlinecite{PauliRevisited}], if its 1-body reduced density matrix satisfies
\begin{equation}
\label{LMEdefeqn}
\gamma_1^\Psi=N\Tr_{2\dots N}[\ketbra{\Psi}]=\frac{N}{d}\cdot\mathds{1}_d.
\end{equation}
These states form a subset $V^{N,d}_{\LME}\subset \wedge^N\mathbb{C}^d$.
\end{definition}

It turns out fermionic LME states exist for almost all $N$ and $d$.
\begin{theorem}[\text{Altunbulak--Klyachko [\onlinecite{PauliRevisited}]}]
\label{LMEtable}
Fermionic LME states exist unless 
\[
d\geq2, N=1\hspace{1.2cm}d \text{\normalfont{ odd}}, N=2\hspace{1.2cm}d \text{\normalfont{ odd}}, N=d-2\hspace{1.2cm}d\geq2, N=d-1
\]
\end{theorem}

Table \ref{LMEfermions1} illustrates this. Note that particle-hole symmetry is present because  $\gamma^{\Psi_{\text{holes}}}_1=\mathds{1}_d-\gamma^{\Psi_{\text{particles}}}_1$ for particle-hole duals $\ket{\Psi_{\text{particles}}}\in\wedge^{N}\mathbb{C}^{d}$ and $\ket{\Psi_{\text{holes}}}\in\wedge^{d-N}\mathbb{C}^{d}$, so that the LME property is preserved. 

\begin{remark}
Though it is not needed in this paper, the dimension of $V^{N,d}_{\LME}/SU(d)$ can be computed with techniques from [\onlinecite{andreev1967},\onlinecite{Raamsdonk2},\onlinecite{Raamsdonk1},\onlinecite{elashvili1972},\onlinecite{kempf1979length}]. For completeness, we include a theorem in Appendix \ref{dimLME}.
\end{remark}

\savebox{\tempbox}{\begin{tabular}{@{}r@{}l@{\space}}
&$\boldsymbol{d}$\\$\boldsymbol{N}$
\end{tabular}}

\begin{center}
\begin{table}
\begin{tabular}{c || P{0.7cm}| P{0.7cm} |P{0.7cm} |P{0.7cm} |P{0.7cm} |P{0.7cm}|P{0.7cm} |P{0.7cm} |P{0.7cm} |P{0.7cm} |P{0.7cm} |P{0.7cm} }
\tikz[overlay]{\draw (1pt,\ht\tempbox) -- (\wd\tempbox,-\dp\tempbox);}%
\usebox{\tempbox}

 &\multicolumn{1}{c}{$\boldsymbol{1}$} &\multicolumn{1}{c}{$\boldsymbol{2}$}&\multicolumn{1}{c}{$\boldsymbol{3}$}&\multicolumn{1}{c}{$\boldsymbol{4}$}&\multicolumn{1}{c}{$\boldsymbol{5}$}&\multicolumn{1}{c}{$\boldsymbol{6}$}&\multicolumn{1}{c}{$\boldsymbol{7}$}&\multicolumn{1}{c}{$\boldsymbol{8}$}&\multicolumn{1}{c}{$\boldsymbol{9}$}&\multicolumn{1}{c}{$\boldsymbol{10}$}&\multicolumn{1}{c}{$\boldsymbol{11}$}&\multicolumn{1}{c}{$\boldsymbol{12}$}\\ 
\toprule
$\boldsymbol{0}$&$\blue{\boldsymbol{\large \checkmark}}$&$\blue{\boldsymbol{\large \checkmark}}$&$\blue{\boldsymbol{\large \checkmark}}$&$\blue{\boldsymbol{\large \checkmark}}$&$\blue{\boldsymbol{\large \checkmark}}$&$\blue{\boldsymbol{\large \checkmark}}$&$\blue{\boldsymbol{\large \checkmark}}$&$\blue{\boldsymbol{\large \checkmark}}$&$\blue{\boldsymbol{\large \checkmark}}$&$\blue{\boldsymbol{\large \checkmark}}$&$\blue{\boldsymbol{\large \checkmark}}$&$\blue{\boldsymbol{\large \checkmark}}$\\
\cline{2-13}
 $\boldsymbol{1}$&$\blue{\boldsymbol{\large \checkmark}}$ & $\red{\boldsymbol{\large \cross}}$&$ \red{\boldsymbol{\large \cross}}$& $\red{\boldsymbol{\large \cross}}$& $\red{\boldsymbol{\large \cross}}$& $\red{\boldsymbol{\large \cross}}$& $\red{\boldsymbol{\large \cross}}$& $\red{\boldsymbol{\large \cross}}$& $\red{\boldsymbol{\large \cross}}$& $\red{\boldsymbol{\large \cross}}$& $\red{\boldsymbol{\large \cross}}$& $\red{\boldsymbol{\large \cross}}$\\
\cline{2-13}
  $\boldsymbol{2}$& &$\blue{\boldsymbol{\large \checkmark}}$  & $\red{\boldsymbol{\large \cross}}$&$\blue{\boldsymbol{\large \checkmark}}$& $\red{\boldsymbol{\large \cross}}$ &$\blue{\boldsymbol{\large \checkmark}}$& $\red{\boldsymbol{\large \cross}}$ &$\blue{\boldsymbol{\large \checkmark}}$& $\red{\boldsymbol{\large \cross}}$ &$\blue{\boldsymbol{\large \checkmark}}$& $\red{\boldsymbol{\large \cross}}$ &$\blue{\boldsymbol{\large \checkmark}}$\\
  
\cline{3-13} \clineB{10-13}{3} \clineB{7-10}{3}
 $\boldsymbol{3}$&\multicolumn{1}{c}{} & & $\blue{\boldsymbol{\large \checkmark}}$ & $\red{\boldsymbol{\large \cross}}$ &  $\red{\boldsymbol{\large \cross}}$ &\multicolumn{1}{a{0.25mm}}{$\blue{\boldsymbol{\large \checkmark}}$}&$\blue{\boldsymbol{\large \checkmark}}$ &$\blue{\boldsymbol{\large \checkmark}}$&$\blue{\boldsymbol{\large \checkmark}}$&$\blue{\boldsymbol{\large \checkmark}}$&$\blue{\boldsymbol{\large \checkmark}}$&$\blue{\boldsymbol{\large \checkmark}}$\\
\cline{4-13} \clineB{7-7}{3}
 $\boldsymbol{4}$&\multicolumn{1}{c}{} &\multicolumn{1}{c}{} & \multicolumn{1}{c}{}& $\blue{\boldsymbol{\large \checkmark}}$  & $\red{\boldsymbol{\large \cross}}$ &$\blue{\boldsymbol{\large \checkmark}}$ & \multicolumn{1}{a{0.25mm}}{ $\blue{\boldsymbol{\large \checkmark}}$}&$\blue{\boldsymbol{\large \checkmark}}$&$\blue{\boldsymbol{\large \checkmark}}$&$\blue{\boldsymbol{\large \checkmark}}$&$\blue{\boldsymbol{\large \checkmark}}$&$\blue{\boldsymbol{\large \checkmark}}$\\
\cline{5-13} \clineB{8-8}{3}
 $\boldsymbol{5}$& \multicolumn{1}{c}{}&\multicolumn{1}{c}{} &\multicolumn{1}{c}{} & &$\blue{\boldsymbol{\large \checkmark}}$ & $\red{\boldsymbol{\large \cross}}$& $\red{\boldsymbol{\large \cross}}$ &\multicolumn{1}{a{0.25mm}}{$\blue{\boldsymbol{\large \checkmark}}$} &$\blue{\boldsymbol{\large \checkmark}}$&$\blue{\boldsymbol{\large \checkmark}}$&$\blue{\boldsymbol{\large \checkmark}}$&$\blue{\boldsymbol{\large \checkmark}}$\\
\cline{6-13}  \clineB{9-9}{3}
 $\boldsymbol{6}$&\multicolumn{1}{c}{} &\multicolumn{1}{c}{} & \multicolumn{1}{c}{}&\multicolumn{1}{c}{} & &$\blue{\boldsymbol{\large \checkmark}}$ &$\red{\boldsymbol{\large \cross}}$ &$\blue{\boldsymbol{\large \checkmark}}$& \multicolumn{1}{a{0.25mm}}{$\blue{\boldsymbol{\large \checkmark}}$}&$\blue{\boldsymbol{\large \checkmark}}$&$\blue{\boldsymbol{\large \checkmark}}$&$\blue{\boldsymbol{\large \checkmark}}$\\
\cline{7-13} \clineB{10-10}{3} 
 $\boldsymbol{7}$&\multicolumn{1}{c}{}&\multicolumn{1}{c}{} &\multicolumn{1}{c}{} & \multicolumn{1}{c}{}&\multicolumn{1}{c}{} & &$\blue{\boldsymbol{\large \checkmark}}$ & $\red{\boldsymbol{\large \cross}}$ & $\red{\boldsymbol{\large \cross}}$ &\multicolumn{1}{a{0.25mm}}{ $\blue{\boldsymbol{\large \checkmark}}$} &$\blue{\boldsymbol{\large \checkmark}}$&$\blue{\boldsymbol{\large \checkmark}}$\\
\cline{8-13} \clineB{11-11}{3}
 $\boldsymbol{8}$&\multicolumn{1}{c}{} &\multicolumn{1}{c}{} &\multicolumn{1}{c}{} &\multicolumn{1}{c}{} &\multicolumn{1}{c}{} &\multicolumn{1}{c}{} & &$\blue{\boldsymbol{\large \checkmark}}$ & $\red{\boldsymbol{\large \cross}}$ &$\blue{\boldsymbol{\large \checkmark}}$  &\multicolumn{1}{a{0.25mm}}{$\blue{\boldsymbol{\large \checkmark}}$} &$\blue{\boldsymbol{\large \checkmark}}$\\
\cline{9-13} \clineB{12-12}{3}
 $\boldsymbol{9}$&\multicolumn{1}{c}{} &\multicolumn{1}{c}{} &\multicolumn{1}{c}{} &\multicolumn{1}{c}{} &\multicolumn{1}{c}{} &\multicolumn{1}{c}{} &\multicolumn{1}{c}{} & &$\blue{\boldsymbol{\large \checkmark}}$  & $\red{\boldsymbol{\large \cross}}$ & $\red{\boldsymbol{\large \cross}}$&\multicolumn{1}{b{0.25mm}}{$\blue{\boldsymbol{\large \checkmark}}$}\\
\cline{10-13} \clineB{13-13}{3}
 $\boldsymbol{10}$&\multicolumn{1}{c}{} &\multicolumn{1}{c}{} &\multicolumn{1}{c}{} &\multicolumn{1}{c}{} &\multicolumn{1}{c}{} &\multicolumn{1}{c}{} &\multicolumn{1}{c}{} &\multicolumn{1}{c}{} & &$\blue{\boldsymbol{\large \checkmark}}$  & $\red{\boldsymbol{\large \cross}}$ &$\blue{\boldsymbol{\large \checkmark}}$ \\
\cline{11-13}
 $\boldsymbol{11}$&\multicolumn{1}{c}{} &\multicolumn{1}{c}{} &\multicolumn{1}{c}{} &\multicolumn{1}{c}{} &\multicolumn{1}{c}{} &\multicolumn{1}{c}{} &\multicolumn{1}{c}{} &\multicolumn{1}{c}{} &\multicolumn{1}{c}{} & &$\blue{\boldsymbol{\large \checkmark}} $ &$\red{\boldsymbol{\large \cross}}$\\
\cline{12-13}
 $\boldsymbol{12}$&\multicolumn{1}{c}{} &\multicolumn{1}{c}{} &\multicolumn{1}{c}{} &\multicolumn{1}{c}{} &\multicolumn{1}{c}{} &\multicolumn{1}{c}{} &\multicolumn{1}{c}{} &\multicolumn{1}{c}{} &\multicolumn{1}{c}{} &\multicolumn{1}{c}{} & &$\blue{\boldsymbol{\large \checkmark}}$ \\
\cline{13-13}
\end{tabular}
\caption{Existence of LME states in $\wedge^N\mathbb{C}^d$. No LME states exist for $(d\geq2, N=1)$, $(d\text{ odd},N=2)$ and their particle-hole duals $(d\geq 2,N=d-1)$, $(d\text{ odd},N=d-2)$.}
\label{LMEfermions1}
\end{table}
\end{center}

From Theorem \ref{LMEtable} and Proposition \ref{extremepointtoLME}, we can now tell which extreme points \eqref{extremefermions} of $P_{d,N}$ are in $F_{d,N}$: each extreme point is indexed by $(i,j)$ and corresponds to a different box of Table \ref{LMEfermions1}. Table \ref{LMEfermions2} illustrates this for $d=11$, $N=5$.

\savebox{\tempbox}{\begin{tabular}{@{}r@{}l@{\space}}
&$\boldsymbol{d}$\\$\ \ \boldsymbol{N}$
\end{tabular}}
\begin{table}
\centering
\begin{tabular}{c || P{0.9cm}| P{0.9cm} |P{0.9cm} |P{0.9cm} |P{0.9cm} |P{0.9cm}|P{0.9cm} |P{0.9cm} |P{0.9cm} |P{0.9cm} |P{0.9cm} |P{0.9cm}}
\tikz[overlay]{\draw (1pt,\ht\tempbox) -- (\wd\tempbox,-\dp\tempbox);}%
\usebox{\tempbox}

 &\multicolumn{1}{c}{$\boldsymbol{1}$} &\multicolumn{1}{c}{$\boldsymbol{2}$}&\multicolumn{1}{c}{$\boldsymbol{3}$}&\multicolumn{1}{c}{$\boldsymbol{4}$}&\multicolumn{1}{c}{$\boldsymbol{5}$}&\multicolumn{1}{c}{$\boldsymbol{6}$}&\multicolumn{1}{c}{$\boldsymbol{7}$}&\multicolumn{1}{c}{$\boldsymbol{8}$}&\multicolumn{1}{c}{$\boldsymbol{9}$}&\multicolumn{1}{c}{$\boldsymbol{10}$}&\multicolumn{1}{c}{$\boldsymbol{11}$}&\multicolumn{1}{c}{$\boldsymbol{12}$}\\ 
\toprule
$\boldsymbol{0}$&  & & & & & & & & & & & \\
\cline{2-13}
 $\boldsymbol{1}$& & \red{$(4,5)$} & \red{$(4,4)$} & \red{$(4,3)$} & \red{$(4,2)$} & \red{$(4,1)$} & \red{$(4,0)$} & & &  &  & \\
\cline{2-13}
 $\boldsymbol{2}$& &  & \red{$(3,5)$} & $(3,4)$ & \red{$(3,3)$} & $(3,2)$ & \red{$(3,1)$} & $(3,0)$  & & & &  \\
\cline{3-13} 
 $\boldsymbol{3}$&\multicolumn{1}{c}{} & &  & \red{$(2,5)$} &  \red{$(2,4)$} & $(2,3)$ & $(2,2)$ & $(2,1)$ & $(2,0)$ & & &\\
\cline{4-13} 
 $\boldsymbol{4}$&\multicolumn{1}{c}{} &\multicolumn{1}{c}{} & &  & \red{$(1,5)$} & $(1,4)$
 & $(1,3)$ & $(1,2)$ & $(1,1)$ & $(1,0)$ & &\\
\cline{5-13} 
 $\boldsymbol{5}$ & \multicolumn{1}{c}{}&\multicolumn{1}{c}{} &\multicolumn{1}{c}{} & & & \red{$(0,5)$} & \red{$(0,4)$} & $(0,3)$ & $(0,2)$ & $(0,1)$ & $(0,0)$ &\\
\cline{6-13}  
 $\boldsymbol{6}$&\multicolumn{1}{c}{} &\multicolumn{1}{c}{} & \multicolumn{1}{c}{}&\multicolumn{1}{c}{} & & & & & & & &
\end{tabular}
\caption{The $30$ extreme points \eqref{extremefermions} of $P_{11,5}$ that are not the Slater point can be associated with the filled boxes. Each extreme point is indexed by $(i,j)$, with $i=$ number of ones, and $j=$ number of zeros. The red points are not in $F_{11,5}$.}
\label{LMEfermions2}
\end{table}

This observation leads to the following conclusion: as $d$ grows, more and more extreme points of $P_{d,N}$ correspond to blue boxes in Table \ref{LMEfermions1}---that is, they are in $F_{d,N}$. The points that $F_{d,N}$ does not reach effectively correspond to $N=1,2,d-1,d-2$. Note that these points have `few-body' character.

As mentioned in Section \ref{proofstrategy}, we will have to approach these problematic points to capture a large volume. That is, we will seek points in $F_{d,N}$ that are fairly close to the problematic points. Lemma \ref{interpolate} shows which (suboptimal) points we use. It is interesting to note that these again have few-body characteristics, in the sense that they consist of a Slater determinant and two constituent parts that correspond to $N=3,4,5$ states or their particle-hole duals. All this supports the idea that the problematic parts of $P_{d,N}$ somehow relate to few-body states---antisymmetry is most restrictive at low particle numbers and the non-trivial Pauli constraints quantify this. 

The following remark makes this a little more precise. 

\begin{remark}
In Section \ref{AcontainedinFsubsection}, we show that for $1\leq m\leq N-7$ and $t=\frac{N-m+1}{N-m+9}$, 
\begin{equation}
A_{d,N,m,t}:=\Big\{(\lambda_1,\dots,\lambda_d)\in P_{d,N}\left|\ \lambda_m\leq t\right.\Big\}\subset F_{d,N}.
\end{equation}
This means that any point on a non-trivial boundary of $F_{d,N}$ needs to have $\lambda_m\geq \frac{N-m+1}{N-m+9}$ for $1\leq m\leq N-7$. For example for $N=1000$, this implies that $\lambda_{209}\geq0.99$, that $\lambda_{609}\geq0.98$, etc. For large $N$, this shows that a state on a non-trivial boundary of $F_{d,N}$ has a dominant Slater determinantal part. Based on numerics (inspired by [\onlinecite{walteralgorithm0}]; now see  [\onlinecite{walteralgorithm}]), we expect that sharper bounds can be found, which could mean that even states with $N=O(100)$ have an approximate Slater determinantal part if they lie on a non-trivial boundary of $F_{d,N}$.
\end{remark}

\subsection{Discussion and outlook}
\label{physimplications}
Many suggested applications of the Pauli constraints involve non-trivial boundaries of $F_{d,N}$ (e.g.\ [\onlinecite{benavides2013quasipinning},\onlinecite{klyachko2009},\onlinecite{schilling2015hubbard},\onlinecite{schilling2015quasipinning},\onlinecite{schilling2018generalized},\onlinecite{schilling2019implications}]). This paper provides some guidance on where these boundaries are, and clarifies which extreme points of $P_{d,N}$  cannot be reached. As discussed above, the problematic extreme points relate to 1 and 2-particle (or hole) states, and the non-trivial boundaries seem to be in the neighbourhood of these points. 

Of course, volume convergence does not mean the Pauli constraints cannot play a role in nature. Effective few-fermion states appear in atoms and also Cooper pairs; many ground states involve correlated pockets with only a few electrons. In general, it remains unclear whether near-Slater determinant ground states of many-electron systems have the tendency to lie close to non-trivial boundaries of $F_{d,N}$. To decide if this is the case, it would be good to study more specific examples, notably ones with more electrons than those considered in [\onlinecite{schilling2013pinning},\onlinecite{tennie2016pinning},\onlinecite{tennie2017influence}]. Another open problem is the implication of our results for Reduced Density Matrix Functional Theory (RDMFT).

Since physical systems often involve spin, let us add a final remark about the spin-dependent polytopes discussed in [\onlinecite{PauliRevisited}]. The analysis and methods used here  extend easily to that case, with similar conclusions.

\section{Estimates and proofs}
\label{proofs}

\subsection{Geometry of $B_{d,N}$, $P_{d,N}$}
\label{geometry}

It will be convenient to gather some facts about polytopes before we start. 

\begin{definition}
A convex polytope is an intersection of a finite number of half-spaces. It can therefore be characterized as the points $(x_1,\dots,x_d)\in\mathbb{R}^d$ that satisfy a finite system of linear inequalities $Ax\leq b$, $A:\mathbb{R}^d\to\mathbb{R}^k$, or
\begin{equation}
\label{polyeqns}
\begin{aligned}
A_{11}x_1+\dots+A_{1d}x_d&\leq b_1\\
\vdots\hspace{1.2cm}&\\
A_{k1}x_1+\dots+A_{kd}x_d&\leq b_k.
\end{aligned}
\end{equation}
\end{definition}

An extreme point of a set $P$ is a point $x\in P$ that cannot be written as a convex combination of two points in $P$ that are distinct from $x$. It is that standard fact that the extreme points of $P$ can be characterized with the equations \eqref{polyeqns}.

\begin{lemma}
\label{extremecharacterization}
Given a polytope $P$ defined by $k$ equations \eqref{polyeqns}, the extreme points of $P$ are those points in $P$ that satisfy $d$ linearly independent equations with equality.
\end{lemma}
\begin{proof}
We study the two inclusions separately.

1.\ Assume $k\geq d$ and that a point $x\in P$ satisfies $d$ linearly independent equations of \eqref{polyeqns} with equality. Also suppose  $x=\mu y+(1-\mu)z$ for $y,z\in P$, $\mu\in[0,1]$. Gather the satisfied equations in $\tilde{A}:\mathbb{R}^d\to\mathbb{R}^d$, and $\tilde{b}\in\mathbb{R}^d$. Since $\tilde{A}y,\tilde{A}z\leq \tilde{b}$ and $\tilde{A}x=\mu \tilde{A}y+(1-\mu)\tilde{A}z=\tilde{b}$, we have $\tilde{A}y=\tilde{A}z=\tilde{b}$, but such a system of $d$ linearly independent equations can have at most one solution, so $x=y=z$.

2.\ Suppose a point $x\in P$ does not satisfy $d$ linearly independent equations of \eqref{polyeqns} with equality. We want to prove that it can be written as $x=\mu y+(1-\mu)z$ with $y,z \in P$ distinct from $x$. We will do this by finding $v\in\mathbb{R}^d$ and $\epsilon>0$ such that $A(x+\epsilon v), A(x-\epsilon v)\leq b$, so that $x+\epsilon v, x-\epsilon v\in P$. The existence of such a $v$ is obvious if $\ker(A)\neq\{0\}$, so assume $\ker(A)=\{0\}$, which implies $k\geq d$ and $\rank(A)=d$. By our assumption, there are at least $k-d+1$ equations in \eqref{polyeqns} that are strict inequalities, and the corresponding basis vectors define a $(k-d+1)$-dimensional subspace of $\mathbb{R}^k$. If we can find $v\in\mathbb{R}^d$ such that $Av$ lies completely in that subspace, there exists $\epsilon>0$ such that $A(x+\epsilon v), A(x-\epsilon v)\leq b$. But such a $v$ exists, since we have assumed that the image of $A$ is $d$-dimensional, and so intersects any $(k-d+1)$-dimensional subspace of $\mathbb{R}^k$ in some non-zero point.
\end{proof}

The convex polytopes we will study are all closed and bounded. In this case the Krein--Milman theorem says that they are in fact the convex hull of their extreme points. The minimal such $d$-dimensional object is a $d$-dimensional simplex---a convex hull of $d+1$ linearly independent points.

The bosonic polytope $B_{d,N}$ is a simplex. For completeness, we discuss it before turning to $P_{d,N}$.

\begin{proposition}[Properties of $B_{d,N}$]
\label{bosontope}
The extreme points of $B_{d,N}$ \eqref{Boseset} are
\begin{equation}
\label{extremebosons}
\begin{tabular}{ c  } 
$(N,\underbrace{0,\dots,0}_{d-1})$\\
$(\frac{N}{2},\frac{N}{2},\underbrace{0,\dots,0}_{d-2})$\\
\vdots\\
$(\frac{N}{d},\dots\dots,\frac{N}{d})$,
\end{tabular}
\end{equation}
and so $B_{d,N}$ is a $(d-1)$-dimensional simplex. It has volume
\begin{equation}
\Vol_{d-1}(B_{d,N})=\frac{\sqrt{d}}{d!}\frac{1}{(d-1)!}\ N^{d-1}.
\end{equation}
\end{proposition}
\begin{proof}
According to Lemma \ref{extremecharacterization}, any extreme point has to satisfy $d$ linearly independent defining equations with equality. There are $d$ inequalities $\lambda_1\geq\lambda_2\geq\dots\geq\lambda_d\geq0$ and one equality $\lambda_1+\dots+\lambda_d=N$. Hence an extreme point is obtained when we ignore one inequality from the list, solve the system of equations, and find that the solution lies in $B_{d,N}$. This gives the $d$ extreme points \eqref{extremebosons}.

To calculate the volume, note that the set 
\begin{equation}
\label{unord}
B^{\text{unord}}_{d,N}:=\left\{(\lambda_1,\dots,\lambda_d)\in\mathbb{R}^d\left|\ \text{$\lambda_i\geq0$ and $\lambda_1+\dots+\lambda_d=N$}\right.\right\}
\end{equation}
can be split in $d!$ pieces of equal volume based on the ordering of the $\lambda_i$, one of which is $B_{d,N}$. Hence,
\begin{equation}
\label{unordered}
\Vol^{d-1}(B_{d,N})=\frac{1}{d!}\Vol^{d-1}(B^{\text{unord}}_{d,N}).
\end{equation}
Also note that $B^{\text{unord}}_{d,N}$ has $d$ linearly independent extreme points $(N,0,\dots,0),\dots,(0,\dots,0,N)$, and so is a $(d-1)$-dimensional simplex. In fact, it is a regular simplex since all its points are equally spaced. If we add the origin $(0,\dots,0)$, we obtain a $d$-dimensional simplex whose volume is easily calculated to be $N^d/d!$ using the standard volume formula for cones (base$\cdot$height$/$dimension). As illustrated in Fig.\ \ref{derivative}, the $(d-1)$-dimensional volume of $B^{\text{unord}}_{d,N}$ can then be found be travelling distance $\epsilon$ in the normal direction $(1,\dots,1)/\sqrt{d}$ and noting that now $\lambda_1+\dots\lambda_d=N+\sqrt{d}\epsilon$, so that the new volume is $(N+\sqrt{d}\epsilon)^d/d!$. Taking a derivative with respect to $\epsilon$ and combining this with \eqref{unordered} gives the volume of $B_{d,N}$.\footnote{A more reliable, but less intuitive way to calculate the volume is to use the parametrization $(\lambda_1,\dots,\lambda_{d-1},N-\lambda_1-\dots-\lambda_{d-1})$, where $\sqrt{d}$ arises as the determinant of the Jacobian.}
\end{proof}

\begin{figure}
\includegraphics[scale=1.2]{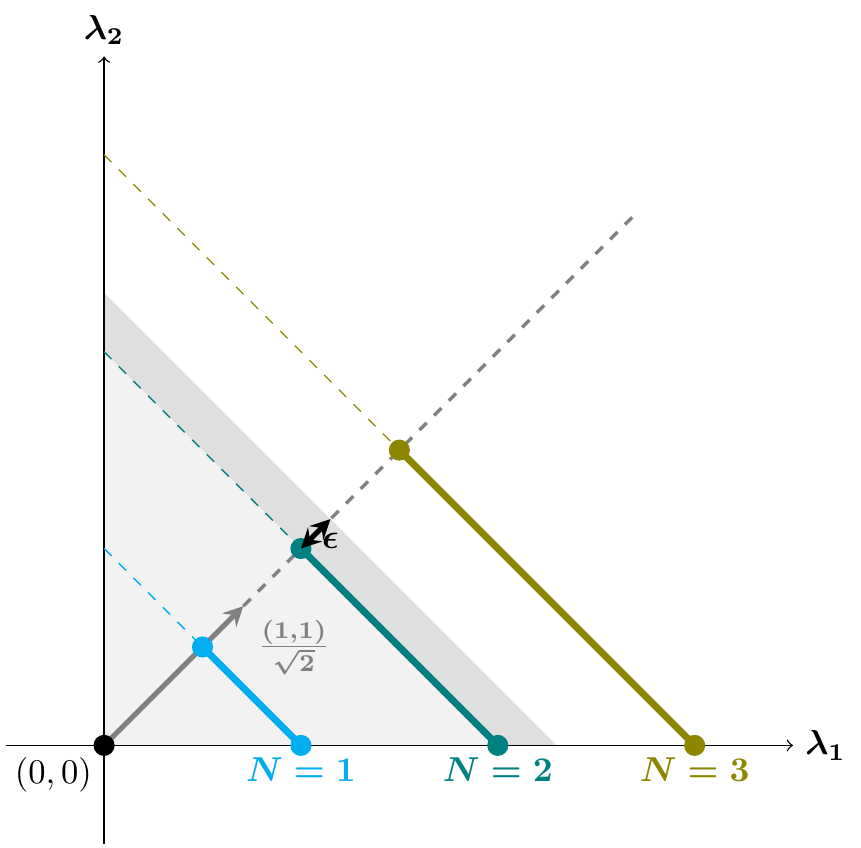}
\caption{Bosonic polytopes $B_{2,N}$ for $N=\lambda_1+\lambda_2=1,2,3$ (bold). The unordered polytopes $B^{\text{unord}}_{2,N}$ are $d!=2$ times as large and include the dashed continuations. The $B^{\text{unord}}_{2,N}$ have length $\sqrt{2}N$. One way to calculate this is to note that the lightly shaded region (drawn for $N=2$) has area $N^2/2$. This increases to $(N+\sqrt{2}\epsilon)^2/2$ if we move a distance $\epsilon$ in the normal direction $(1,1)/\sqrt{2}$. Hence the derivative in $\epsilon$ at $\epsilon=0$, which is the surface length, is $\sqrt{2}N$.}
\label{derivative}
\end{figure}

\begin{proof}[Proof of Proposition \ref{paulipolytope}: Properties of $P_{d,N}$]
Using Lemma \ref{extremecharacterization} again, this time there are $d+1$ inequalities $1\geq\lambda_1\geq\lambda_2\geq\dots\geq\lambda_d\geq0$, and one equality $\lambda_1+\dots+\lambda_d=N$. Ignoring two inequalities from the list, solving the resulting set of $d$ equations, and checking whether solutions lie in $P_{d,N}$ results in \eqref{extremefermions}. The extreme points are completely defined by the number of $\lambda_i$ equal to 1 ($\leq N$) and 0 ($\leq d-N$).

To calculate the volume, we use the same method as for $B_{d,N}$. To deal with the ordering, define
\begin{equation}
P^{\text{unord}}_{d,N}:=\left\{(\lambda_1,\dots,\lambda_d)\in\mathbb{R}^d\left|\ \text{$1\geq\lambda_i\geq0$ and $\lambda_1+\dots+\lambda_d=N$}\right.\right\}.
\end{equation}

\begin{figure}
\includegraphics[scale=1.2]{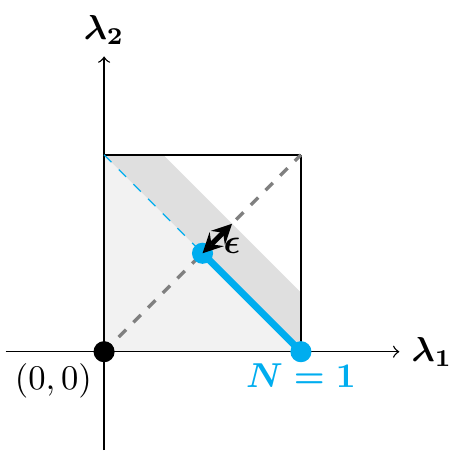}
\caption{The Pauli polytope $P_{2,1}$ (bold). The unordered polytope $P^{\text{unord}}_{2,1}$ is $d!=2$ times as large and includes the dashed continuation. The method to determine the volume still applies, and we still have $\lambda_1+\lambda_2=N+\sqrt{2}\epsilon$ upon moving distance $\epsilon$ in the normal direction. Naturally, the area of the shaded region is subject to the Pauli bound $\lambda_i\leq1$.}
\label{derivativefermi}
\end{figure}

As illustrated in Fig.\ \ref{derivativefermi}, the bound $\lambda_i\leq1$ complicates the volume of this object. It is now convenient to use the Irwin-Hall distribution of probability theory. For uniformly distributed i.i.d.\ random variables and $x\in\mathbb{R}$,
\begin{equation}
\label{IrwinHallprob}
\mathbb{P}_{X_i\sim \U(0,1)}[X_1+\dots+X_d\leq x]=\frac{1}{d!}\sum^{\lfloor x\rfloor}_{k=0}(-1)^k\binom{d}{k}(x-k)^d.
\end{equation}
Note this is exactly the volume of the convex set generated by $P^{\text{unord}}_{d,N}$ and the origin $(0,\dots,0)$ (see [\onlinecite{marengo2017geometric}] for a review). Hence, as before,
\begin{equation}
\begin{aligned}
\Vol^{d-1}(P^{\text{unord}}_{d,N})&=\left.\partial_\epsilon\mathbb{P}_{X_i\sim \U(0,1)}[X_1+\dots+X_d\leq N+\sqrt{d}\epsilon]\right|_{\epsilon=0}\\
&=\frac{\sqrt{d}}{(d-1)!}\sum^{N}_{k=0}(-1)^k\binom{d}{k}(N-k)^{d-1}.
\end{aligned}
\end{equation}
$\Vol^{d-1}(P_{d,N})$ acquires an extra $1/d!$ because of ordering.
\end{proof}

\subsection{Fermionic states: Proof of Proposition \ref{extremepointtoLME}}
\label{fermionicstates}
To prove this proposition, we need to review some properties of fermionic states. 
First of all, an $N$-fermion \emph{Slater determinant} built from orthonormal $\ket{u_1},\dots,\ket{u_N}\in\mathbb{C}^d$ is defined as
\begin{equation}
\label{Slater}
\ket{u_1\wedge\dots\wedge u_N}:=\frac{1}{\sqrt{N!}}\sum_{\sigma\in S_N}\sgn(\sigma)\ket{u_{\sigma(1)}\otimes\dots\otimes u_{\sigma(N)}}.
\end{equation}
This satisfies antisymmetry under permutations $\sigma\in S_N$, or
\begin{equation}
\ket{u_{1}\wedge\dots\wedge u_{N}}=\sgn(\sigma)\ket{u_{\sigma(1)}\wedge\dots\wedge u_{\sigma(N)}}.
\end{equation}
For an orthonormal basis $\ket{u_1},\dots,\ket{u_d}$ of $\mathbb{C}^d$, the $\binom{d}{N}$ Slater determinants built from that basis are an orthonormal basis of $\wedge^N\mathbb{C}^d$. For a state $\ket{\Phi}\in\wedge^N\mathbb{C}^d$ whose expansion in this basis does not involve Slater determinants containing $\ket{u_i}$, we define $\ket{u_i\wedge\Phi}\wedge^{N+1}\mathbb{C}^d$ by linearity. For example,
\begin{equation}
\label{wedge}
\ket{u_1}\wedge\left(\frac{1}{\sqrt{2}}\ket{u_2\wedge u_3}+\frac{1}{\sqrt{2}}\ket{u_4\wedge u_5}\right)=\frac{1}{\sqrt{2}}\ket{u_1\wedge u_2\wedge u_3}+\frac{1}{\sqrt{2}}\ket{u_1\wedge u_4\wedge u_5}.
\end{equation}

We will extend the definition of $\wedge$ somewhat in Lemma \ref{combining}. To do this, define the projection onto $\wedge^N\mathbb{C}^d\subset\otimes^N\mathbb{C}^d$ by
\begin{equation}
   \Pi^N_A:=\frac{1}{N!}\sum_{\sigma\in S_N}\sgn(\sigma)U_\sigma
\end{equation}
where $U_\sigma$ is the permutation operator corresponding to $\sigma$. Comparing this to the definition of a Slater determinant \eqref{Slater}, we note
\begin{equation}
\label{someref8382}
    \ket{v_1\wedge\dots\wedge v_N}=\sqrt{N!}\ \Pi^N_A \ket{v_1\otimes\dots\otimes v_N}
\end{equation}
Finally, recall that the annihilation operator $a_i$ corresponding to $\ket{u_i}$ acts as
\begin{equation}
\label{annihilationdef}
a_i=\sqrt{N}(\bra{u_i}\otimes\mathds{1}):\wedge^N\mathbb{C}^d\to\wedge^{N-1}\mathbb{C}^d.
\end{equation}
This implies that
\begin{equation}
\label{containing2}
a_i\ket{u_{j_1}\wedge \dots\wedge u_{j_k}\wedge u_i\wedge u_{j_{k+1}}\wedge\dots\wedge u_{j_{N-1}}}=(-1)^{k}\ket{u_{j_1}\wedge\dots\wedge u_{j_{N-1}}},
\end{equation}
and also that $a_i$ gives zero on Slaters that do not contain $\ket{u_i}$. Consequently, splitting an $N$-fermion state $\ket{\Psi}=\ket{\Psi_1}+\ket{\Psi_2}$ into a part $\ket{\Psi_1}$ containing Slaters without $\ket{u_i}$, and a part $\ket{\Psi_2}=\ket{u_i\wedge\Phi}$, we obtain
\begin{equation}
\label{containing}
a_i\ket{\Psi}=\ket{\Phi}.
\end{equation}

We are now ready to prove Proposition \ref{extremepointtoLME} with the following two lemmas.
\begin{lemma}
\label{cuttingaway}
Let $\ket\Psi\in\wedge^N\mathbb{C}^d$ be normalized with a 1-body reduced density matrix $\gamma^\Psi_1$ that has ordered eigenvalues $(\lambda^\Psi_1,\dots,\lambda^\Psi_{d})$ and corresponding eigenbasis $\ket{u_1},\dots,\ket{u_d}$. 

1. if $\lambda^\Psi_d=0$, then $\ket{\Psi}$ can be expanded in Slaters not containing $\ket{u_d}$, and hence be embedded in $\wedge^N\mathbb{C}^{d-1}$.

2. if $\lambda^\Psi_1=1$, then $\ket\Psi=\ket{u_1\wedge\Phi}$, and $\ket\Phi\in\wedge^{N-1}\mathbb{C}^{d-1}$ can be expanded in Slaters not containing $\ket{u_1}$.
\end{lemma}
\begin{proof}
1. Using \eqref{annihilationdef}, we find that the norm of the part of $\ket\Psi$ that contains $\ket{u_d}$ is 
\begin{equation}
\|a_d\ket\Psi\|^2=\bra{u_d}\gamma^\Psi_1\ket{u_d}=0,
\end{equation}
so according to \eqref{containing}, no Slater in $\ket\Psi$ contains $\ket{u_d}$.

2. Similarly, the norm of the part of $\ket\Psi$ containing $\ket{u_1}$ is 
\begin{equation}
\|a_1\ket\Psi\|^2=\bra{u_1}\gamma^\Psi_1\ket{u_1}=1,
\end{equation}
so all Slaters in this basis contain $\ket{u_1}$, and we can write $\ket\Psi=\ket{u_1\wedge\Phi}$ and $a_1\ket\Psi=\ket\Phi$.
\end{proof}

\begin{lemma}
\label{combining}
For $i=1,2$, suppose that $\ket{\Psi_i}\in\wedge^{N_i}\mathbb{C}^{d_i}$ has 1-body reduced density matrix $\gamma^{\Psi_i}_1$ with eigenvalues $\lambda^{(i)}_j$, $1\leq j\leq d_i$ and corresponding eigenvectors $\ket{u^{(i)}_j}$, such that $\ket{u^{(1)}_j}$ and $\ket{u^{(2)}_{j'}}$ are mutually orthogonal for all $j,j'$. Then, extend \eqref{wedge} by
\begin{equation}
\label{someref2380}
\ket{\Psi_1\wedge\Psi_2}:=\sqrt{\tbinom{N_1+N_2}{N_1}}\ \Pi^{N_1+N_2}_A\ket{\Psi_1\otimes\Psi_2}\in\wedge^{N_1+N_2}\mathbb{C}^{d_1+d_2}.
\end{equation}
This state is normalized and its 1-body reduced density matrix is $\gamma^{\Psi_1}_1+\gamma^{\Psi_2}_1$ with eigenvalues $\{\lambda^{(i)}_j|\ 1\leq j\leq d_i,\ i=1,2\}$.
\end{lemma}
\begin{proof}
Using \eqref{someref8382} and the projection property $\Pi^{N_1+N_2}_A(\Pi^{N_1}_A\otimes \Pi^{N_2}_A)=\Pi^{N_1+N_2}_A$, it is easy to see that for $v_1\dots v_{N_1+N_2}$ orthonormal,
\begin{equation}
\sqrt{\tbinom{N_1+N_2}{N_1}}\Pi^{N_1+N_2}\ket{v_1\wedge\dots\wedge v_{N_1}\otimes v_{N_1+1}\wedge\dots\wedge v_{N_1+N_2}}=\ket{v_1\wedge\dots\wedge v_{N_1+N_2}}
\end{equation}
is normalized. By linearity this directly extends to $\ket{\Psi_1\otimes\Psi_2}$. 

To show the eigenvalue property, denote sets of $N_1$ distinct vectors $\{u^{(1)}_{j_1},\dots,u^{(1)}_{j_{N_1}}\}$ by $S$, and their corresponding (ordered) Slater determinant by $\ket{S}$. Then,
\begin{equation}
   \ket{\Psi_1}=\sum_Sc_S\ket{S}
\end{equation}
for suitable coefficients $c_S$ with $\sum_S|c_S|^2=1$. By \eqref{annihilationdef} and \eqref{containing2}, this implies
\begin{equation}
\label{someref12345}
\begin{aligned}
    \lambda^{(1)}_j\delta_{jj'}=\bra{u^{(1)}_j}\gamma^{\Psi_1}_1\ket{u^{(1)}_{j'}}&=\bra{\Psi_1}(a^{(1)}_{j'})^\dagger a^{(1)}_j\ket{\Psi_1}\\&=\sum_{S\ni j,S'\ni j'}\sgn(u^{(1)}_j,S)\sgn(u^{(1)}_{j'},S')\overline{c_{S'}}c_{S}\braket{S'\backslash u^{(1)}_{j'}|S\backslash u^{(1)}_{j}},
\end{aligned}
\end{equation}
where $\sgn(u^{(1)}_j,S)$ is the sign of the permutation that reorders the elements of $S$ from increasing to $u^{(1)}_j$ first, then increasing. Note that the inner product is 1 if $S\backslash u^{(1)}_{j}=S'\backslash u^{(1)}_{j'}$ and 0 otherwise.

Adopting a similar notation for $\ket{\Psi_2}$ with an index $T$, we find
\begin{equation}
    \ket{\Psi_1\wedge\Psi_2}=\sum_{S,T}c_Sc_T\ket{S\cup T},
\end{equation}
noting $S\cap T=\emptyset$ for all $S,T$. It is then easy to see that cross terms
\begin{equation}
    \bra{u^{(1)}_j}\gamma^{\Psi_1\wedge\Psi_2}_1\ket{u^{(2)}_{j'}}=\bra{\Psi_1\wedge\Psi_2}(a^{(2)}_{j'})^\dagger a^{(1)}_j\ket{\Psi_1\wedge\Psi_2}=0
\end{equation}
by orthogonality. Also, 
\begin{equation}
\begin{aligned}
    \bra{u^{(1)}_j}&\gamma^{\Psi_1\wedge\Psi_2}_1\ket{u^{(1)}_{j'}}=\bra{\Psi_1\wedge\Psi_2}(a^{(1)}_{j'})^\dagger a^{(1)}_j\ket{\Psi_1\wedge\Psi_2}\\
    &=\sum_{S,S',T}\sgn(u^{(1)}_j,S)\sgn(u^{(1)}_{j'},S')\ \overline{c_{S'}}c_{S}|c_T|^2\braket{S'\backslash u^{(1)}_{j'}\cup T|S\backslash u^{(1)}_{j}\cup T}=\lambda^{(1)}_j\delta_{j j'},
\end{aligned}
\end{equation}
so $\lambda^{(1)}_j$ is indeed an eigenvalue of $\gamma_1^{\Psi_1\wedge\Psi_2}$. The same argument applies to the $\lambda^{(2)}_j$. 
\end{proof}

\subsection{Proving $A_{d,N,m,t}\subset F_{d,N}$ for certain $m,t$}
\label{AcontainedinFsubsection}
Recall that $A_{d,N,m,t}$ was defined for integers $1\leq m\leq d$ and $t\in[0,1]$ as
\begin{equation}
A_{d,N,m,t}:=\Big\{(\lambda_1,\dots,\lambda_d)\in\mathbb{R}^d\left|\ \text{$1\geq\lambda_1\geq\dots\geq\lambda_d\geq0$ and $\lambda_m\leq t$ and $\sum^d_{i=1}\lambda_i=N$}\right.\Big\}.
\end{equation}

\begin{proposition}
\label{AcontainedinF}
Let $8\leq N\leq d/2$, $m\leq N-7$ and $t=\frac{N-m+1}{N-m+9}$. Then, $A_{d,N,m,t}\subset F_{d,N}$.
\end{proposition}

Note that $A_{d,N,m,t}$ is a polytope. Our goal is to show that all of its extreme points are contained in $F_{d,N}$. Recall from Lemma \ref{extremecharacterization} that the extreme points of a polytope in $\mathbb{R}^d$ satisfy $d$ of the polytope's defining equations. Since $A_{d,N,m,t}$ is $P_{d,N}$ constrained by $\lambda_m\leq t$, its extreme points come in two types:
\begin{itemize}
    \item \textit{Extreme points of $P_{d,N}$ \eqref{extremefermions2} satisfying $\lambda_m\leq t$.} These satisfy $d-1$ equations of $1\geq\lambda_1\geq\dots\geq\lambda_d\geq0$ with equality, and also $\sum_i\lambda_i=N$.
    \item \textit{Extreme points with $\lambda_m=t$}. In addition, these satisfy $d-2$ equations of $1\geq\lambda_1\geq\dots\geq\lambda_d\geq0$ with equality, as well as $\sum_i\lambda_i=N$.
\end{itemize}

We need to check that both types are contained in $F_{d,N}$. For the first type, recall the extreme points of $P_{d,N}$ from Proposition \ref{paulipolytope}. These were the Slater point $(1,\dots,1,0,\dots,0)$ and 
\begin{equation}
\label{extremefermions2}
\big(\underbrace{\vphantom{\frac{N-i}{d-i-j}}1,\dots,1}_{i},\underbrace{\frac{N-i}{d-i-j},\dots,\frac{N-i}{d-i-j}}_{d-i-j},\underbrace{\vphantom{\frac{N-i}{d-i-j}}0,\dots,0}_{j}\big)\in\mathbb{R}^d \hspace{0.7cm} (0\leq i\leq N-1,\ 0\leq j\leq d-N-1).
\end{equation}
Also recall from Section \ref{extremepointssection} that these are definitely in $F_{d,N}$ unless $N-i=1,2$ or $N-i=d-i-j-2,d-i-j-1$. The following lemma now shows the condition $\lambda_m\leq t$ excludes these problematic cases.
\begin{lemma}
\label{insideA}
Let $8\leq N\leq d/2$, $m\leq N-7$ and $t=\frac{N-m+1}{N-m+9}$. Then, of all the extreme points of $P_{d,N}$, only those in \eqref{extremefermions2} with $0\leq i\leq m-1$ and $0\leq j\leq d-i-\frac{N-i}{t}$ are contained in $A_{d,N,m,t}$. In particular, points with $i>N-8$ and $j>d-N-8$ lie outside of $A_{d,N,m,t}$.
\end{lemma}
\begin{proof}
Since $t<1$ and $m<N$, the Slater point and all points with $i\geq m$ are excluded. The points with $i\leq m-1$ have $\lambda_m=\frac{N-i}{d-i-j}$ since $d-j\geq N+1$. These points are only included if $\frac{N-i}{d-i-j}\leq t$, which gives the equation for $j$. The final remark follows from $i\leq m-1\leq N-8$ and $j\leq d-i-\frac{N-i}{t}\leq d-m+1-\frac{N-m+1}{t}=d-N-8$.
\end{proof}
It is now clear that $A_{d,N,m,t}$'s extreme points of the first type are in $F_{d,N}$, but we are not ready for a proof of Proposition \ref{AcontainedinF}. We also need to study extreme points of the second type, namely those with $\lambda_m=t$. We will actually ignore this defining property, and focus on the fact they satisfy $d-2$ equations of $1\geq\lambda_1\geq\dots\geq\lambda_d\geq0$ with equality instead.

\begin{lemma}
\label{d-2eqns}
Let $(\lambda_1,\dots,\lambda_d)\in P_{d,N}$, and assume that at least $d-2$ equations of $1\geq\lambda_1\geq\dots\geq\lambda_d\geq0$ are satisfied with equality. Then, $(\lambda_1,\dots,\lambda_d)$ can be written as a convex combination of (at most) two extreme points of $P_{d,N}$ that satisfy the same $d-2$ equations with equality.
\end{lemma}
\begin{proof}
Note that the $d-2$ equalities, together with $\sum_i\lambda_i=N$, define a 1-dimensional subspace of $\mathbb{R}^d$. Our point must lie in the intersection of $P_{d,N}$ with this subspace, which is a bounded convex set with a most two extreme points. It is defined by the above $d-1$ equalities, together with the two remaining inequalities of $1\geq\lambda_1\geq\dots\geq\lambda_d\geq0$. According to Lemma \ref{extremecharacterization}, its extreme points satisfy one of those inequalities with equality, and hence $d$ of the defining equations of $P_{d,N}$: they are extreme points of $P_{d,N}$.
\end{proof}

This says we will not miss out on any extreme points of $A_{d,N,m_0,t_0}$ if we restrict attention to line segments between extreme points of $P_{d,N}$ that share $d-2$ equalities of $1\geq\lambda_1\geq\dots\geq\lambda_d\geq0$. Fortunately, we know many such line segments are completely contained in $F_{d,N}$, simply because their defining extreme points are, according to the analysis from Section \ref{extremepointssection}. We will need more information for line segments whose endpoints are not both contained in $F_{d,N}$. It turns out the following lemma provides this, as will be explained in the proof of Proposition \ref{AcontainedinF} further down.

\begin{lemma}
\label{interpolate}
Let $8\leq N\leq d/2$, $m\leq N-7$ and $t=\frac{N-m+1}{N-m+9}$. Consider an extreme point \eqref{extremefermions2} of $P_{d,N}$ indexed by $(i,j)$ with $0\leq i\leq m-1$ and $0\leq j\leq d-i-\frac{N-i}{t}$. Then, the following points are contained in $F_{d,N}$.
\begin{equation}
\label{interpoints}
\begin{aligned}
&\big(\underbrace{\vphantom{\frac{4}{d-N-j+1}}1,\dots,1}_{i},\underbrace{\vphantom{\frac{4}{d-N-j+1}}\frac{N-i-4}{N-i-1},\dots,\frac{N-i-4}{N-i-1}}_{N-i-1},\underbrace{\frac{4}{d-N-j+1},\dots,\frac{4}{d-N-j+1}}_{d-N-j+1},\underbrace{\vphantom{\frac{4}{d-N-j+1}}0,\dots,0}_{j}\big)\\
&\big(\underbrace{\vphantom{\frac{5}{d-N-j+2}}1,\dots,1}_{i},\underbrace{\vphantom{\frac{5}{d-N-j+2}}\frac{N-i-5}{N-i-2},\dots,\frac{N-i-5}{N-i-2}}_{N-i-2},\underbrace{\frac{5}{d-N-j+2},\dots,\frac{5}{d-N-j+2}}_{d-N-j+2},\underbrace{\vphantom{\frac{5}{d-N-j+2}}0,\dots,0}_{j}\big)\\
&\big(\underbrace{\vphantom{\frac{3}{d-N-j-1}}1,\dots,1}_{i},\underbrace{\vphantom{\frac{3}{d-N-j-1}}\frac{N-i-3}{N-i+1},\dots,\frac{N-i-3}{N-i+1}}_{N-i+1},\underbrace{\frac{3}{d-N-j-1},\dots,\frac{3}{d-N-j-1}}_{d-N-j-1},\underbrace{\vphantom{\frac{3}{d-N-j-1}}0,\dots,0}_{j}\big)\\
&\big(\underbrace{\vphantom{\frac{3}{d-N-j-2}}1,\dots,1}_{i},\underbrace{\vphantom{\frac{3}{d-N-j-2}}\frac{N-i-3}{N-i+2},\dots,\frac{N-i-3}{N-i+2}}_{N-i+2},\underbrace{\frac{3}{d-N-j-2},\dots,\frac{3}{d-N-j-2}}_{d-N-j-2},\underbrace{\vphantom{\frac{3}{d-N-j-2}}0,\dots,0}_{j}\big)
\end{aligned}
\end{equation}
These points lie on the line segments between the extreme point of $P_{d,N}$ indexed by $(i,j)$ and those indexed by $(N-1,j)$, $(N-2,j)$, $(i,d-N-1)$ and $(i,d-N-2)$ respectively. Any extreme points of $A_{d,N,m,t}$ on these line segments are contained in $F_{d,N}$.
\end{lemma}
\begin{proof}
We only discuss the first point of \eqref{interpoints}. The others can be treated in a similar way. 

To prove that the point is contained in $F_{d,N}$, we concatenate LME states using Lemma \ref{combining}. To start, take an $i$-dimensional subspace of $\mathbb{C}^d$ and construct a Slater determinant $\ket{\Psi_1}$. We then pick an $(N-i-1)$-dimensional subspace of the remaining $\mathbb{C}^{d-i}$ and construct an LME state $\ket{\Psi_2}$ of $N-i-4$ particles, which exists by Theorem \ref{LMEtable} since $N-i-4\geq 4$. Finally, we pick an $(d-N-j+1)$-dimensional subspace of the remaining $\mathbb{C}^{d-N+1}$ and construct an LME state $\ket{\Psi_3}$ of $4$ particles. which exists since $d-N-j+1\geq 9$. Lemma \ref{combining} then says that $\ket{\Psi_1\wedge\Psi_2\wedge\Psi_3}\in\wedge^N\mathbb{C}^d$ with the desired (ordered) eigenvalue vector. Hence, the point is in $F_{d,N}$.

Now consider the statement about the line segment. It is easy to see that the three points are on a line. Their order is also simple to check, for instance in the first case by verifying $\frac{N-i-4}{N-i-1}\in[\frac{N-i}{d-i-j},1]$ using $\frac{N-i}{d-i-j}\leq \frac{N-i}{N-i+8}\leq\frac{N-i-4}{N-i-1}$ since $j\leq d-N-8$ as used in Lemma \ref{insideA}.

For the final statement, note that the extreme point of $P_{d,N}$ indexed by $(i,j)$ is in $F_{d,N}$, and that it has $\lambda_{m}=\frac{N-i}{d-i-j}\leq t$ by Lemma \ref{insideA}. The first point of \eqref{interpoints} is also in $F_{d,N}$, but it has $\lambda_{m}=\frac{N-i-4}{N-i-1}\geq\frac{N-m-3}{N-m}\geq\frac{N-m+1}{N-m+9}=t$ by our assumptions.
Since $\lambda_m$ is strictly increasing on the line segment between $(i,j)$ and $(N-1,j)$, this means that all points with $\lambda_m\leq t$ on that line segment are in $F_{d,N}$, but then so must any extreme points of $A_{d,N,m,t}$ be.
\end{proof}

We are now ready to prove that $A_{d,N,m,t}\subset F_{d,N}$ when $m\leq N-7$ and $t=\frac{N-m+1}{N-m+9}$.

\begin{proof}[Proof of Proposition \ref{AcontainedinF}]
Recall that we wanted to show that all extreme points of $A_{d,N,m,t}$ are in $F_{d,N}$. We identified two types of extreme points below Proposition \ref{AcontainedinF}: points that are also extreme points of $P_{d,N}$, and points that are not, but satisfy $\lambda_m=t$. Lemma \ref{insideA} says that points of the first type are all contained in $F_{d,N}$.

For points of the second type, Lemma \ref{d-2eqns} proves that we can restrict our attention to line segments between certain pairs of extreme points of $P_{d,N}$. In many cases such line segments are entirely in $F_{d,N}$ since their endpoints are according to Theorem \ref{LMEtable} and Proposition \ref{extremepointtoLME}.

Which pairs of extreme points are left to check? We claim that the pairs addressed in Lemma \ref{interpolate} suffice. Why? By our reasoning just now, any remaining pairs must contain a member that corresponds to one of the problematic extreme points of Theorem \ref{LMEtable}, which leaves only the cases $i=N-1,N-2$ or $j=d-N-1,d-N-2$. These points are outside $A_{d,N,m,t}$ by Lemma \ref{insideA}, so the other member of the pair should be inside---the interpolation could not contain an extreme point of $A_{d,N,m,t}$ otherwise. Also, Lemma \ref{d-2eqns} says that it suffices to consider two points that have $d-2$ equalities of $1\geq\lambda_1\geq\dots\geq\lambda_d\geq0$ in common. Each point \eqref{extremefermions2} satisfies $d-1$ of these with equality, so is easy to see that only pairs $(i,j)$, $(i',j')$ with $i=i'$ or $j=j'$ qualify---see Table \ref{LMEfermions3} for the position of such pairs in the LME table. 

All these considerations reduce our efforts to exactly the pairs discussed in Lemma \ref{interpolate}. That lemma also showed that any extreme points of $A_{d,N,m,t}$ on the corresponding line segments are in $F_{d,N}$, so that all extreme points of $A_{d,N,m,t}$ are, and indeed the set itself is. 
\end{proof}

\savebox{\tempbox}{\begin{tabular}{@{}r@{}l@{\space}}
&$\boldsymbol{d}$\\$\ \ \boldsymbol{N}$
\end{tabular}}
\begin{table}
\centering
\begin{tabular}{c || P{1cm}| P{1cm} |P{1cm} |P{1cm} |P{1cm} |P{1cm}|P{1cm} |P{1cm} |P{1cm} |P{1cm} |P{1cm} |P{1cm}}
\tikz[overlay]{\draw (1pt,\ht\tempbox) -- (\wd\tempbox,-\dp\tempbox);}%
\usebox{\tempbox}

 &\multicolumn{1}{c}{$\boldsymbol{1}$} &\multicolumn{1}{c}{$\boldsymbol{2}$}&\multicolumn{1}{c}{$\boldsymbol{3}$}&\multicolumn{1}{c}{$\boldsymbol{4}$}&\multicolumn{1}{c}{$\boldsymbol{5}$}&\multicolumn{1}{c}{$\boldsymbol{6}$}&\multicolumn{1}{c}{$\boldsymbol{7}$}&\multicolumn{1}{c}{$\boldsymbol{8}$}&\multicolumn{1}{c}{$\boldsymbol{9}$}&\multicolumn{1}{c}{$\boldsymbol{10}$}&\multicolumn{1}{c}{$\boldsymbol{11}$}&\multicolumn{1}{c}{$\boldsymbol{12}$}\\ 
\toprule
$\boldsymbol{0}$&  & & & & & & & & & & & \\
\cline{2-13}
 $\boldsymbol{1}$&  & \red{$(4,5)$} & \red{$(4,4)$} & \red{$(4,3)$} & \red{$(4,2)$} & \red{$\boldsymbol{(4,1)}$} & \red{$(4,0)$} & & &  &  & \\
\cline{2-13}
 $\boldsymbol{2}$& & & \red{$(3,5)$} & \blue{$(3,4)$} & \red{$(3,3)$} & \blue{$(3,2)$} & \red{$\boldsymbol{(3,1)}$} & \blue{$(3,0)$}  & & & &  \\
\cline{3-13} 
 $\boldsymbol{3}$&\multicolumn{1}{c}{} & & & \red{$(2,5)$} &  \red{$(2,4)$} & \blue{$(2,3)$} & \blue{$(2,2)$} & $\blue{\boldsymbol{(2,1)}}$ & \blue{$(2,0)$} & & &\\
\cline{4-13} 
 $\boldsymbol{4}$&\multicolumn{1}{c}{} &\multicolumn{1}{c}{} & & & \red{$\boldsymbol{(1,5)}$} & \blue{$\boldsymbol{(1,4)}$}
 & \blue{$\boldsymbol{(1,3)}$} & \blue{$\boldsymbol{(1,2)}$} & $\boldsymbol{(1,1)}$ & \blue{$\boldsymbol{(1,0)}$} & &\\
\cline{5-13} 
 $\boldsymbol{5}$ & \multicolumn{1}{c}{}&\multicolumn{1}{c}{} &\multicolumn{1}{c}{} & & & \red{$(0,5)$} & \red{$(0,4)$} & \blue{$(0,3)$} & \blue{$(0,2)$} & \blue{$\boldsymbol{(0,1)}$} & \blue{$(0,0)$} &\\
\cline{6-13}  
 $\boldsymbol{6}$&\multicolumn{1}{c}{} &\multicolumn{1}{c}{} & \multicolumn{1}{c}{}&\multicolumn{1}{c}{} & & & & & & & &
\end{tabular}
\caption{Illustration of the proof of Proposition \ref{AcontainedinF}, following the $\wedge^5\mathbb{C}^{11}$ example of Table \ref{LMEfermions2}. Consider pairs involving the point $(1,1)$. According to Lemma \ref{d-2eqns}, it suffices to consider the bold points, because we can connect $(1,1)$ to these points by a line that satisfies $d-2$ of $1\geq\lambda_1\geq\dots\geq\lambda_d$ with equality. Of the bold points, only the line segments connecting to the red points are not automatically contained in $F_{d,N}$ by Theorem \ref{LMEtable}, so these require the additional work from Lemma \ref{interpolate}.
Note, though, that this illustration is not perfect: the case $\wedge^5\mathbb{C}^{11}$ is not actually covered by the main theorem---this example is just explain these considerations.}
\label{LMEfermions3}
\end{table}

\subsection{Volume estimates}
\label{volumeestimates}
The important conclusion from Proposition \ref{AcontainedinF} is that for certain $m$ and $t$,
\begin{equation}
\frac{\Vol^{d-1}(F_{d,N})}{\Vol^{d-1}(P_{d,N})}\geq\frac{\Vol^{d-1}(A_{d,N,m,t})}{\Vol^{d-1}(P_{d,N})}.
\end{equation}
We now start estimating this ratio.

First note that we can remove the ordering by adding a factor $1/d!$ to both volumes, and replacing $\lambda_m\leq t$ by $\lambda_{[m]}\leq t$, where the latter denotes the  $m$th largest entry of the tuple $(\lambda_1,\dots,\lambda_d)$. This implies
\begin{equation}
\label{tobound2}
\begin{aligned}
&\frac{\Vol^{d-1}(A_{d,N,m,t})}{\Vol^{d-1}(P_{d,N})}=\frac{\Vol^{d-1}\Big(\big\{(\lambda_1,\dots,\lambda_d)\in\mathbb{R}^d\left|\ \text{$\lambda_i\in[0,1]$ and $\lambda_{[m]}\leq t$ and $\sum_{i=1}^d\lambda_i=N$}\right.\big\}\Big)}{\Vol^{d-1}\Big(\big\{(\lambda_1,\dots,\lambda_d)\in\mathbb{R}^d\left|\ \text{$\lambda_i\in[0,1]$ and $\sum_{i=1}^d\lambda_i=N$}\right.\big\}\Big)}\\
    &\hspace{1.5cm}=1-\frac{\Vol^{d-1}\Big(\big\{(\lambda_1,\dots,\lambda_d)\in\mathbb{R}^d\left|\ \text{$\lambda_i\in[0,1]$ and $\lambda_{[m]}> t$ and $\sum_{i=1}^d\lambda_i=N$}\right.\big\}\Big)}{\Vol^{d-1}\Big(\big\{(\lambda_1,\dots,\lambda_d)\in\mathbb{R}^d\left|\ \text{$\lambda_i\in[0,1]$ and $\sum_{i=1}^d\lambda_i=N$}\right.\big\}\Big)}\\
    &\hspace{1.5cm}\geq 1-\frac{\Vol^{d-1}\Big(\big\{(\lambda_1,\dots,\lambda_d)\in\mathbb{R}^d\left|\ \text{$\lambda_{[m]}> t$ and $\sum_{i=1}^d\lambda_i=N$}\right.\big\}\Big)}{\Vol^{d-1}\Big(\big\{(\lambda_1,\dots,\lambda_d)\in\mathbb{R}^d\left|\ \text{$\lambda_i\in[0,1]$ and $\sum_{i=1}^d\lambda_i=N$}\right.\big\}\Big)}\\
    &\hspace{1.5cm}\geq 1-\binom{d}{m}\frac{\Vol^{d-1}\Big(\big\{(\lambda_1,\dots,\lambda_d)\in\mathbb{R}^d\left|\ \text{$\lambda_1,\dots,\lambda_m>t$ and $\sum_{i=1}^d\lambda_i=N$}\right.\big\}\Big)}{\Vol^{d-1}\Big(\big\{(\lambda_1,\dots,\lambda_d)\in\mathbb{R}^d\left|\ \text{$\lambda_i\in[0,1]$ and $\sum_{i=1}^d\lambda_i=N$}\right.\big\}\Big)},
    \end{aligned}
\end{equation}
where we use permutation invariance in the last step (see Figure \ref{illustration} for a geometric example). 

\begin{figure}
        \centering
        \includegraphics[scale=0.37]{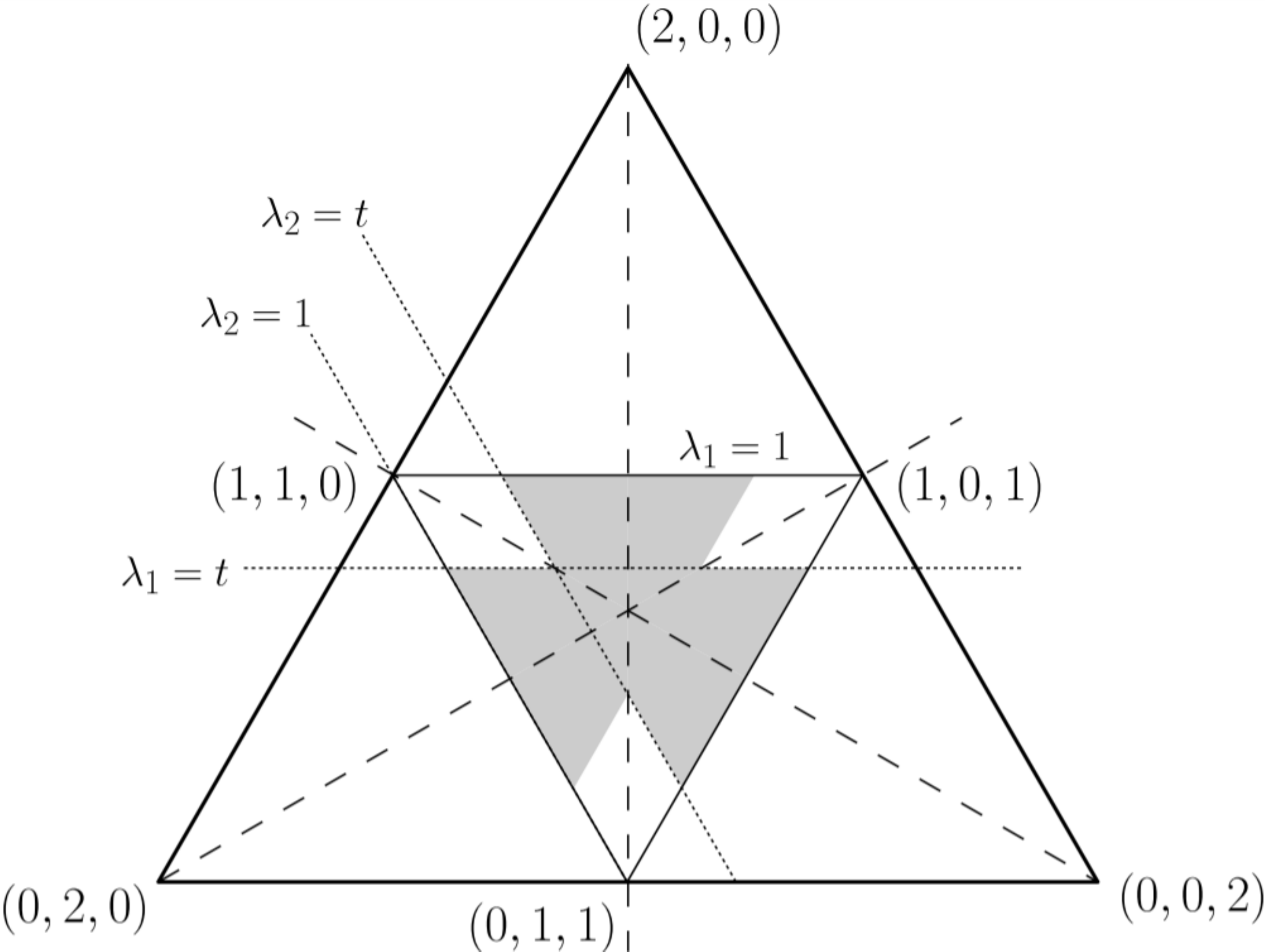}
        \includegraphics[scale=0.58]{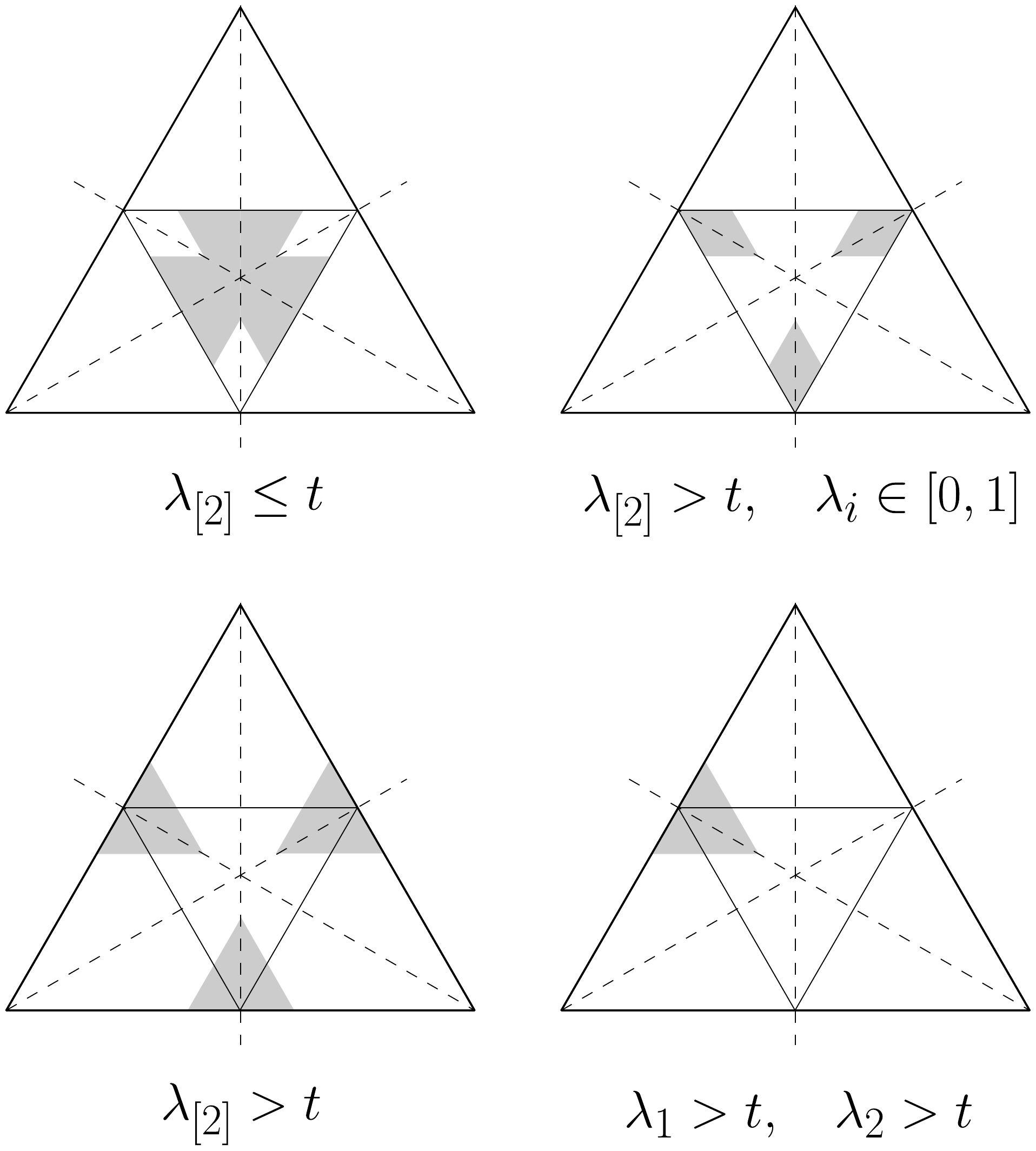}
        \caption{This image illustrates \eqref{tobound2} for $d=3$ and $N=2$. Starting with the image at the top, the large triangle represents $B_{3,2}$, containing $P_{3,2}$ as a smaller triangle. The area defined by $\lambda_{[2]}\leq t$ is indicated in grey in the first two images. We then proceed through the steps of \eqref{tobound2} image by image. These drawings were kindly contributed by an anonymous referee.}
        \label{illustration}
\end{figure}

In Proposition \ref{paulipolytope}, we showed the volume in the denominator is equal to
\begin{equation}
\label{someref29939}
    \sqrt{d}\left.\partial_y\mathbb{P}
    _{X_i\sim\U(0,1)}\Big[\sum^d_{i=0}X_i\leq y\Big]\right|_N=\sqrt{d}\frac{1}{(d-1)!}\sum^{\lfloor
    N\rfloor}_{k=0}(-1)^k\binom{d}{k}(N-k)^{d-1},
\end{equation}
which is the probability density function of the Irwin-Hall distribution. To give a lower bound on \eqref{tobound2}, we need a lower bound on this quantity, but for fixed $N$ and $d\to\infty$ that amounts to a large deviations estimate.
The only exception is $N=d/2$, so we aim to reduce to that case by proving the following lemma.

\begin{lemma}
\label{lemma103}
Assuming $d\geq 3$, the quantity
\begin{equation}
\label{goal9320}
    \sqrt{d}\left.\partial_y\mathbb{P}_{X_i\sim\U(0,1)}\Big[\sum^d_{i=0}X_i\leq y\Big]\right|_x x^{-(d-1)}
\end{equation}
is continuously differentiable and monotone decreasing in $x\in(0,\infty)$.
\end{lemma}
\begin{proof}
Since the derivative above is the probability density function of a sum of i.i.d.\ uniform random variables, it is easy to confirm with induction that it is a repeated convolution of the $\U(0,1)$-density $\mathds{1}_{[0,1]}$. That is,
\begin{equation}
\label{someref0123}
\left.\partial_y\mathbb{P}
    _{X_i\sim\U(0,1)}\Big[\sum^d_{i=0}X_i\leq y\Big]\right|_x=\mathds{1}_{[0,1]}^{\ast(d-1)}(x)=\frac{1}{(d-1)!}\sum^{\lfloor x\rfloor}_{k=0}(-1)^k\binom{d}{k}(x-k)^{d-1}.
\end{equation}
This gives continuous differentiability in $x\in(0,\infty)$ for \eqref{goal9320} as long as $d\geq3$. For monotonicity, we use induction. Starting from $d=3$, 
\begin{equation}
\frac{\mathds{1}_{[0,1]}^{\ast2}(x)}{x^2}=\frac{1}{2}\mathds{1}_{[0,\infty)}(x)-\frac{3}{2}(1-\tfrac{1}{x})^2\mathds{1}_{[1,\infty)}(x)+\frac{3}{2}(1-\tfrac{2}{x})^2\mathds{1}_{[2,\infty)}(x)-\frac{1}{2}(1-\tfrac{3}{x})^2\mathds{1}_{[3,\infty)}(x)
\end{equation}
is indeed monotone decreasing on $(0,\infty)$.
Now assume the statement is true for some $d\geq3$, and consider the derivative of \eqref{goal9320} for $d+1$, $x\in(0,\infty)$,
\begin{equation}
   \sqrt{d}\ \frac{x^{d}\partial_y \mathds{1}_{[0,1]}^{\ast d}(x)-dx^{d-1}\mathds{1}_{[0,1]}^{\ast d}(x)}{x^{2d}}.
\end{equation}
We claim this is negative on $(0,\infty)$. Note that induction tells us this is the case for $d$, and so for $x\in(0,\infty)$,
\begin{equation}
(d-1)\mathds{1}_{[0,1]}^{\ast(d-1)}(x)\geq x\partial_y \mathds{1}_{[0,1]}^{\ast(d-1)}(x).
\end{equation}
Adding $\mathds{1}_{[0,1]}^{\ast(d-1)}(x)$ and convoluting with $\mathds{1}_{[0,1]}$ gives
\begin{equation}
\label{calculation2398}
\begin{aligned}
d\mathds{1}_{[0,1]}^{\ast d}(x)&\geq \int^x_{x-1}\mathds{1}_{[0,1]}^{\ast(d-1)}(s)+s\partial_y \mathds{1}_{[0,1]}^{\ast(d-1)}(s)ds=\int^x_{x-1}\partial_y(y\mathds{1}_{[0,1]}^{\ast(d-1)})(s)ds\\
&=x\mathds{1}_{[0,1]}^{\ast(d-1)}(x)-(x-1)\mathds{1}_{[0,1]}^{\ast(d-1)}(x-1)\\
&\geq x\left(\mathds{1}_{[0,1]}^{\ast(d-1)}(x)-\mathds{1}_{[0,1]}^{\ast(d-1)}(x-1)\right)=x\partial_y\mathds{1}_{[0,1]}^{\ast d}(x),
\end{aligned}
\end{equation}
where we used the explicit form of \eqref{someref0123} in last step, and its positivity in the one before. Hence, \eqref{goal9320} is monotone decreasing on $(0,\infty)$.
\end{proof}

This allows us estimate the volume of $P_{d,N}$.
\begin{proposition}
\label{largedeviations}
For $d\geq 7$ and $N\leq d/2$,
\begin{equation}
\begin{aligned}
    d!\Vol^{d-1}(P_{d,N})&=\Vol^{d-1}\Big(\big\{(\lambda_1,\dots,\lambda_d)\in\mathbb{R}^d\Big|\ \text{$\lambda_i\in[0,1]$ \normalfont{and} $\sum_{i=1}^d\lambda_i=N$}\big\}\Big)\\
    &\geq\frac{\sqrt{d}}{2}\left(\frac{2N}{d}\right)^{d-1}.
    \end{aligned}
\end{equation}
\end{proposition}
\begin{proof}
According to Lemma \ref{lemma103} and \eqref{someref0123}, this quantity is lower bounded by 
\begin{equation}
        \sqrt{d}N^{d-1}\left.\partial_y\mathbb{P}_{X_i\sim\U(0,1)}\Big[\sum^d_{i=0}X_i\leq y\Big]\right|_{d/2}\left(\frac{d}{2}\right)^{-(d-1)}=\sqrt{d}\left(\frac{2N}{d}\right)^{d-1}\mathds{1}_{[0,1]}^{\ast(d-1)}(d/2).
\end{equation}
According to the last equality of \eqref{calculation2398} and \eqref{someref0123},
\begin{equation}
\mathds{1}_{[0,1]}^{\ast(d-1)}(d/2)=\int^{d/2}_{-\infty}\big[\mathds{1}_{[0,1]}^{\ast(d-2)}(s)-\mathds{1}_{[0,1]}^{\ast(d-2)}(s-1)\big]ds=\mathbb{P}_{X_i\sim\U(0,1)}\Big[\sum^{d-1}_{i=0}X_i\in[\tfrac{d}{2}-1,\tfrac{d}{2}]\Big].
\end{equation}
By Chebyshev's inequality, this is
\begin{equation}
\mathbb{P}_{X_i\sim\U(0,1)}\Big[\ \Big|\sum^{d-1}_{i=0}X_i-\tfrac{d-1}{2}\Big|\leq\tfrac{1}{2}\Big]\geq 1-\frac{3}{d-1}\geq\frac12
\end{equation}
for $d\geq7$.
\end{proof}

Having dealt with the denominator of \eqref{tobound2}, it remains to calculate the numerator.
\begin{proposition}
\label{numest}
Let $m\geq0$ be an integer and $t\in\mathbb{R}$. Assuming $N\geq mt$,
\begin{equation}
\label{goal212}
    \Vol^{d-1}\Big(\big\{(\lambda_1,\dots,\lambda_d)\in\mathbb{R}^d\Big|\  \text{\normalfont $\lambda_1,\dots,\lambda_m>t$ and $\sum_{i=1}^d\lambda_i=N$}\big\}\Big)=\sqrt{d}\frac{1}{(d-1)!}(N-mt)^{d-1}.
\end{equation}
\end{proposition}
\begin{proof}
In close analogy with \eqref{someref29939}, the volume above is equal to
\begin{equation}
\sqrt{d}\ (2N)^d\left.\partial_y\mathbb{P}
    _{X_i\sim\U(0,2N)}\Big[\text{$X_1,\dots,X_m>t$ and $\sum^d_{i=0}X_i\leq y$}\Big]\right|_N,
\end{equation}
where the value $2N$ was chosen as a convenient number bigger than $N$. Note that the $(2N)^d$ arises as the volume of $[0,2N]^d$. Since this is again a probability density, its value is similar to \eqref{someref0123}, namely
\begin{equation}\sqrt{d}\ (2N)^d \left(\frac{1}{(2N)^m}\mathds{1}^{\ast(m-1)}_{[t,2N]}\ast\frac{1}{(2N)^{d-m}}\mathds{1}^{\ast(d-m-1)}_{[0,2N]}\right)(N)=\sqrt{d} \left(\mathds{1}^{\ast(m-1)}_{[t,2N]}\ast\mathds{1}^{\ast(d-m-1)}_{[0,2N]}\right)(N).
\end{equation}
To show that this is indeed \eqref{goal212}, we use induction on $d,m\to d+1,m+1$ to prove a slightly more general claim, namely that for $x\leq 2N$,
\begin{equation}
   \left(\mathds{1}^{\ast(m-1)}_{[t,2N]}\ast\mathds{1}^{\ast(d-m-1)}_{[0,2N]}\right)(x)=\frac{1}{(d-1)!}(x-mt)^{d-1}\mathds{1}_{[mt,\infty]}(x).
\end{equation}
The base case---$m=0$, any $d$---is covered by the above analysis and \eqref{someref29939}.
We now assume the formula is true for $d,m$ and note
\begin{equation}
    \begin{aligned}
    \mathds{1}_{[t,2N]}\ast\left(\mathds{1}^{\ast(m-2)}_{[t,2N]}\ast\mathds{1}^{\ast(d-m-1)}_{[0,2N]}\right)(x)&=\frac{1}{(d-2)!}\int^{\infty}_{-\infty}(s-mt)^{d-2}\mathds{1}_{[mt,\infty]}(s)\mathds{1}_{[t,2N]}(x-s)ds\\
    &=\frac{1}{(d-2)!}\int^{\max(mt,x-t)}_{mt}(s-mt)^{d-2}ds,
    \end{aligned}
\end{equation}
which proves the claim.
\end{proof}

As a final ingredient, we prove Proposition \ref{Pauliloss} stated in Section \ref{comparison}. 

\begin{proof}[Proof of Proposition \ref{Pauliloss}]
We use techniques mentioned before. Similar to \eqref{tobound2}, we obtain
\begin{equation}
\label{tobound3}
\begin{aligned}
&\frac{\Vol^{d-1}(P_{d,N})}{\Vol^{d-1}(B_{d,N})}=\frac{\Vol^{d-1}\Big(\big\{(\lambda_1,\dots,\lambda_d)\in\mathbb{R}^d\left|\ \text{$\lambda_{[1]}\leq 1$ and $\sum_{i=1}^d\lambda_i=N$}\right.\big\}\Big)}{\Vol^{d-1}\Big(\big\{(\lambda_1,\dots,\lambda_d)\in\mathbb{R}^d\left|\ \text{$\sum_{i=1}^d\lambda_i=N$}\right.\big\}\Big)}\\
    &\hspace{1.5cm}=1-\frac{\Vol^{d-1}\Big(\big\{(\lambda_1,\dots,\lambda_d)\in\mathbb{R}^d\left|\ \text{$\lambda_{[1]}> 1$ and $\sum_{i=1}^d\lambda_i=N$}\right.\big\}\Big)}{\Vol^{d-1}\Big(\big\{(\lambda_1,\dots,\lambda_d)\in\mathbb{R}^d\left|\ \text{ $\sum_{i=1}^d\lambda_i=N$}\right.\big\}\Big)}\\
    &\hspace{1.5cm}\geq 1-d\frac{\Vol^{d-1}\Big(\big\{(\lambda_1,\dots,\lambda_d)\in\mathbb{R}^d\left|\ \text{$\lambda_1>1$ and $\sum_{i=1}^d\lambda_i=N$}\right.\big\}\Big)}{d!\Vol^{d-1}\big(B_{d,N}\big)},
    \end{aligned}
\end{equation}
so the lower bound follows from Proposition \ref{bosontope} and Proposition \ref{numest}. For the upper bound, start again from the middle line of \eqref{tobound3}, use $\lambda_1>1\implies\lambda_{[1]}>1$ and Proposition \ref{numest}.
\end{proof}

The main result now follows by combining the results above.

\begin{proof}[Proof of Theorems \ref{maintheorem} and \ref{fullvsPauli}]
Theorem \ref{maintheorem} follows directly from the bounds of Theorem \ref{fullvsPauli}. These can be derived as follows.

1.\ Fix $8\leq N\leq d/2$, and recall that we previously obtained Proposition \ref{AcontainedinF} and \eqref{tobound2}. Combining this with Proposition \ref{numest} and the lower bound of Proposition \ref{Pauliloss} gives 
\begin{equation}
    \begin{aligned}
\frac{\Vol^{d-1}(F_{d,N})}{\Vol^{d-1}(P_{d,N})}\geq\frac{\Vol^{d-1}(A_{d,N,m,t})}{\Vol^{d-1}(P_{d,N})}&\geq 1-\binom{d}{m}\frac{1}{1-d\big(\frac{N-1}{N}\big)^{d-1}}\left(\frac{N-mt}{N}\right)^{d-1}\\
&\geq 1-\frac{d^N}{1-d\big(\frac{N-1}{N}\big)^{d-1}}\left(\frac{N-mt}{N}\right)^{d-1}
\end{aligned}
\end{equation}
for $m\leq N-7$ and $t=\frac{N-m+1}{N-m+9}$. 

To obtain a good estimate both for low and high $N$, we use two different $m$. The first is simply $m=N-7$, which gives $N-mt=\frac12(N+7)$. The second is $m=N+9-\lceil\sqrt{8}\sqrt{N+9}\rceil$, for which it can be verified that $N-mt\leq \sqrt{32N}$. This bound is not allowed if $N+9-\lceil\sqrt{8}\sqrt{N+9}\rceil\geq N-7$, but in this case  $\min\big[\frac{1}{2}(N+7),\sqrt{32N}\big]=\frac{1}{2}(N+7)$. This proves the estimate.

2. For $N=rd\geq 20$, we again use \eqref{tobound2} and Proposition \ref{numest}, but also Proposition \ref{largedeviations}. This gives
\begin{equation}
\frac{\Vol^{d-1}(F_{d,N})}{\Vol^{d-1}(P_{d,N})}\geq\frac{\Vol^{d-1}(A_{d,N,m,t})}{\Vol^{d-1}(P_{d,N})}\geq 1-2\binom{d}{N}\left(\frac{1}{2r}\right)^{d-1}\frac{1}{(d-1)!}\big(N-mt\big)^{d-1}.
\end{equation}
We choose $m=N+9-\lceil\sqrt{8}\sqrt{N+9}\rceil$ as before, and use $N-mt\leq\sqrt{32N}=\sqrt{32 rd}$.
For the factorials, we use Stirling's formula
\begin{equation}
    \frac{1}{(d-1)!}\leq \frac{1}{\sqrt{2\pi}\sqrt{d-1}}\left(\frac{e}{d-1}\right)^{d-1},
\end{equation}
and\enlargethispage{\baselineskip}
\begin{equation}
\begin{aligned}
    \binom{d}{N}=\binom{d}{rd}=\frac{d!}{(rd)!\big((1-r)d\big)!}&\leq \frac{e d^{d+1/2}e^{-d}}{2\pi (rd)^{rd+1/2}e^{-rd} \big((1-r)d\big)^{(1-r)d+1/2}e^{-(1-r)d}}\\
    &=\frac{e}{2\pi\sqrt{d}}\frac{1}{r^{r+1/2}(1-r)^{3/2-r}}\Big(\frac{1}{r^r(1-r)^{1-r}}\Big)^{d-1}.
\end{aligned}
\end{equation}
Since $d\geq 2N\geq 40$, all this gives
\begin{equation}
\begin{aligned}
\frac{\Vol^{d-1}(F_{d,rd})}{\Vol^{d-1}(P_{d,rd})}&\geq1-2\frac{e}{(2\pi)^{3/2}\sqrt{d(d-1)}}\frac{1}{r^{r+1/2}(1-r)^{3/2-r}}\left(\frac{ed}{2(d-1)}\frac{\sqrt{32}}{r^{r+1/2}(1-r)^{1-r}}\frac{1}{\sqrt{d}}\right)^{d-1}\\
&\geq1-\frac{1}{r^{r+1/2}(1-r)^{3/2-r}}\left(\frac{8}{r^{r+1/2}(1-r)^{1-r}}\frac{1}{\sqrt{d}}\right)^{d-1}.
\end{aligned}
\end{equation}
\end{proof}

\noindent\textbf{Data availability} Data sharing is not applicable to this article as no new data were created or analyzed in this study.

\begin{acknowledgments}
This work was supported by the Royal Society through a Newton International Fellowship, by Darwin College Cambridge through a Schlumberger Research Fellowship, and by the Villum Centre of Excellence QMATH (University of Copenhagen).  I thank Jan Philip Solovej for asking me about the volume of these polytopes, and for his hospitality at QMATH, where this work was initiated. I also thank Michael Walter for discussions about numerical and physical aspects of the problem, and Christian Schilling for comments on the initial draft of this paper. Figure \ref{illustration} was kindly suggested and contributed by an anonymous referee.
\end{acknowledgments}

\appendix
\numberwithin{equation}{section}
\section{Exact volume of $A_{d,N,m,t}$}
\label{VolA}
As discussed in Remark \ref{betterestimate}, the following calculation gives the sharpest estimate our method can produce, but it is not used in the proof of the main theorems. 

First, recall that $A_{d,N,m,t}$ was defined for integers $1\leq m\leq d$, $N\in\mathbb{R}$, $t\in[0,1]$ as
\[
A_{d,N,m,t}:=\Big\{(\lambda_1,\dots,\lambda_d)\in\mathbb{R}^d\left|\ \text{$1\geq\lambda_1\geq\dots\geq\lambda_d\geq0$ and $\lambda_m\leq t$ and $\sum^d_{l=1}\lambda_l=N$}\right.\Big\}.
\]

\begin{theorem}
\label{orderstat}
Let $X_1,\dots,X_d\sim \U(0,1)$ i.i.d.\ and $x\in\mathbb{R}$. For $1\leq m\leq d$, let $X_{(d+1-m)}$ be the $(d+1-m)$th order statistic, that is, the $(d+1-m)$th smallest value, which means it is the $m$th largest value. Then, for $t\in[0,1]$,
\begin{equation}
\label{someref43}
\begin{aligned}
\mathbb{P}\Big[X_{(d+1-m)}\leq t\ &\cap\ \sum^d_{l=1}X_l\leq x\Big]=\frac{1}{d!}\sum^{\min(m-1,\lfloor x\rfloor)}_{i=0}(-1)^i\tbinom{d}{i}(x-i)^{d}\\
&+\frac{1}{d!}\sum^{\min(m-1,\lfloor x\rfloor)}_{i=0}\sum^{\lfloor \frac{x-i}{t}\rfloor}_{k=m-i}(-1)^{k+i}\tbinom{d}{k+i}\tbinom{k+i}{i}\left(\sum^{m-i-1}_{j=0}(-1)^j\tbinom{k}{j}\right)(x-kt-i)^{d}.
\end{aligned}
\end{equation}
Comparable to Proposition \ref{paulipolytope}, this gives the volume
\begin{equation}
    \Vol^{d-1}(A_{d,N,m,t})=\frac{\sqrt{d}}{d!}\partial_x\left.\mathbb{P}\Big[X_{(d+1-m)}\leq t\ \cap\ \sum^d_{l=1}X_l\leq x\Big]\right|_{x=N}.
\end{equation}
\end{theorem}
\begin{remark}
When differentiated in $x$, this probability relates to the order statistics of a bunch of uniform random variables with constraint $\sum_lX_l=x$. Such order statistics are most likely well-known, but we were unable to find a suitable reference.
\end{remark}
Note that by permutation invariance, we have 
\begin{equation}
\label{decompose}
\begin{aligned}
\mathbb{P}\Big[X_{(d+1-m)}\leq t\ \cap\ \sum^d_{l=1}&X_l\leq x\Big]\\
&=\sum^{m-1}_{j=0}\tbinom{d}{j}\mathbb{P}\Big[X_1,\dots,X_j>t,\ X_{j+1},\dots,X_d\leq t\ \cap\ \sum^d_{l=0}X_l\leq x\Big].
\end{aligned}
\end{equation}
We compute the latter probabilities separately. 
\begin{lemma}
\label{someref21}
Let $X_1,\dots,X_d\sim \U(0,1)$ i.i.d.\ and $x\in\mathbb{R}$. For $0\leq j\leq d$, $t\in[0,1]$,
\begin{equation}
\begin{aligned}
\mathbb{P}\Big[X_1,\dots,X_j>t,\ X_{j+1},&\dots,X_d\leq t\ \cap\ \sum^d_{l=0}X_l\leq x\Big]\\
&=\frac{1}{d!}\sum^{j}_{i=0}(-1)^i\tbinom{j}{i}\sum^{\lfloor\frac{x-i}{t}\rfloor-(j-i)}_{k=0}(-1)^k\tbinom{d-j}{k}\big(x-(k+j-i)t-i\big)^{d}
\end{aligned}
\end{equation}
Note that this is zero if $j\geq\lfloor\frac{x}{t}\rfloor+1$.
\end{lemma}
\begin{proof}
We use induction on $j-1,d-1$ to $j,d$, that is, we add a random variable and assume that it is bigger than $t$. The base case has $j=0$ and general $d$, or
\begin{equation}
\mathbb{P}\Big[X_{1},\dots,X_d\leq t\ \cap\ \sum^d_{l=0}X_l\leq x\Big]=\frac{1}{d!}\sum^{\lfloor\frac{x}{t}\rfloor}_{k=0}(-1)^k\tbinom{d}{k}(x-kt)^{d}.
\end{equation}
This can be verified by seeing this probability is equal to
\begin{equation}
\mathbb{P}\Big[\sum^d_{l=0}X_l\leq x\ \big|\ X_{1},\dots,X_d\leq t\Big]\ \mathbb{P}[X_{1},\dots,X_d\leq t]=\mathbb{P}\Big[\sum^d_{l=0}\frac{X_l}{t}\leq\frac{x}{t}\ \big|\ X_{1},\dots,X_d\leq t\Big]t^d,
\end{equation}
and using \eqref{IrwinHallprob}. For the induction step, we integrate over $X_j=s\in[t,1]$. This gives
\begin{equation}
\label{someref20}
\begin{aligned}
&\mathbb{P}\Big[X_1,\dots,X_j>t,\ X_{j+1},\dots,X_d\leq t\ \cap\ \sum^d_{l=0}X_l\leq x\Big]\\
&=\int^1_{t} \mathbb{P}\Big[X_1,\dots,X_{j-1}>t,\ X_{j+1},\dots,X_d\leq t\ \cap\ \sum^d_{l=0}X_l\leq x-s\Big]ds\\
&=\frac{1}{(d-1)!}\sum^{j-1}_{i=0}(-1)^i\tbinom{j-1}{i}\int^1_{t}\sum^{\lfloor\frac{x-s-i}{t}\rfloor-(j-1-i)}_{k=0}(-1)^k\tbinom{d-j}{k}\big(x-s-(k+j-1-i)t-i\big)^{d-1}ds.
\end{aligned}
\end{equation}
Note that for all terms $k\leq\lfloor\frac{x-i-1}{t}\rfloor-(j-i-1)$, the integral is over the entire range $[t,1]$, but that for $k\geq \lfloor\frac{x-i-1}{t}\rfloor-(j-i-1)+1$ it is only over $[t,x-i-(k+j-i-1)t]$. This interval is empty if $k\geq\lfloor\frac{x-i}{t}\rfloor-(j-i)+1$, and so \eqref{someref20} is equal to 
\begin{equation}
\begin{aligned}
\frac{1}{d!}\sum^{j-1}_{i=0}(-1)^i\tbinom{j-1}{i}\Big[&\sum^{\lfloor\frac{x-i}{t}\rfloor-(j-i)}_{k=0}(-1)^k\tbinom{d-j}{k}\big(x-(k+j-i)t-i\big)^{d}\\
&-\sum^{\lfloor\frac{x-i-1}{t}\rfloor-(j-i-1)}_{k=0}(-1)^k\tbinom{d-j}{k}\big(x-(k+j-i-1)t-(i+1)\big)^{d}\Big]\\
&=\frac{1}{d!}\sum^{j}_{i=0}(-1)^i\big[\tbinom{j-1}{i}+\tbinom{j-1}{i-1}\big]\sum^{\lfloor\frac{x-i}{t}\rfloor-(j-i)}_{k=0}(-1)^k\tbinom{d-j}{k}\big(x-(k+j-i)t-i\big)^{d},
\end{aligned}
\end{equation}
which is the desired result.
\end{proof}

\begin{proof}[Proof of Theorem \ref{orderstat}]
Rewriting Lemma \ref{someref21} slightly and checking which terms are clearly zero, we obtain that \eqref{decompose} is equal to 
\begin{equation}
\frac{1}{d!}\sum^{\min(m-1,\lfloor\frac{x}{t}\rfloor)}_{j=0}\tbinom{d}{j}\sum^{\min(j,\lfloor x\rfloor)}_{i=0}\tbinom{j}{i}\sum^{\lfloor\frac{x-i}{t}\rfloor}_{k=j-i}(-1)^{k-j}\tbinom{d-j}{k-(j-i)}\big(x-kt-i\big)^d.
\end{equation}
A careful exchange of the sums gives 
\begin{equation}
\frac{1}{d!}\sum^{\min(m-1,\lfloor x\rfloor)}_{i=0}\sum^{\min(m-1,\lfloor\frac{x}{t}\rfloor)}_{j=i}\sum^{\lfloor\frac{x-i}{t}\rfloor}_{k=j-i}(-1)^{k-j}\tbinom{d}{j}\tbinom{j}{i}\tbinom{d-j}{k-(j-i)}\big(x-kt-i\big)^d.
\end{equation}
We then use $\tbinom{d}{j}\tbinom{j}{i}\tbinom{d-j}{k-(j-i)}=\tbinom{d}{k+i}\tbinom{k+i}{i}\tbinom{k}{j-i}$ and a second exchange of sums to obtain
\begin{equation}
\frac{1}{d!}\sum^{\min(m-1,\lfloor x\rfloor)}_{i=0}\sum^{\lfloor\frac{x-i}{t}\rfloor}_{k=0}(-1)^{k}\tbinom{d}{k+i}\tbinom{k+i}{i}\left(\sum^{\min(m-1,\lfloor\frac{x}{t}\rfloor,i+k)}_{j=i}(-1)^j\tbinom{k}{j-i}\right)\big(x-kt-i\big)^d.
\end{equation}
Notice $k\leq\lfloor\frac{x-i}{t}\rfloor$ implies $k\leq\lfloor\frac{x}{t}\rfloor-i$, and so the part between the big brackets equals
\begin{equation}
(-1)^i\sum^{\min(m-i-1,k)}_{j=0}(-1)^j\tbinom{k}{j},
\end{equation}
which is $(-1)^i$ if $k=0$ and $0$ if $k\leq m-i-1$, so that $k=0$ gives rise to the first term of \eqref{someref43}, and the $k\geq m-i$ to the second. 
\end{proof}

\section{Dimension of $V^{N,d}_{\LME}/\SU(d)$}
\label{dimLME}
For completeness, we extend the results of  [\onlinecite{Raamsdonk2},\onlinecite{Raamsdonk1}] (and the predating qubit case [\onlinecite{macikazek2013many}]) to fermions by calculating the dimension of $V^{N,d}_{\LME}/\SU(d)$ (Definition \ref{LMEdef}). These dimensions are not otherwise used in this paper.
\begin{theorem}
Given that $\SU(d)$ acts as $A\in\SU(d)\mapsto A\otimes\dots\otimes A$ on $V^{N,d}_{\LME}$ (Definition \ref{LMEdef}).
\[
\dim(V^{N,d}_{\LME}/\SU(d))=
\left\{\begin{array}{cl}
    0 &\hspace{1cm} \text{if } N=0, N=d\\
    -1 &\hspace{1cm} \text{if $d\geq 2$ and } N=1, N=d-1\\
    0 &\hspace{1cm} \text{if $d\geq 2$ is even and } N=2, N=d-2\\
    -1 &\hspace{1cm} \text{if $d\geq2$ is odd and } N=2, N=d-2 \\
    \geq0 &\hspace{1cm} \text{if } d=6, N=3 \\  
    \geq0 &\hspace{1cm} \text{if $d=7$ and } N=3, N=4 \\
    \geq0 &\hspace{1cm} \text{if $d=8$ and } N=3, N=5 \\
    \binom{d}{N}-d^2 &\hspace{1cm} \text{if $d=8,N=4$ or $d\geq9$ and } 3\leq N\leq d-3\\
        \end{array}\right.
\]
Here, dimension $-1$ indicates $V^{N,d}_{\LME}/\SU(d)=\emptyset$, whereas dimension 0 indicates that it is a point. The orange results only indicate existence of LME states.
\end{theorem}

\begin{proof}
\textit{$N=0, N=d$}\\
This case is trivial since there is only one normalized state and it satisfies \eqref{LMEdefeqn}. 
\item \textit{$d\geq2$ and $N=1, N=d-1$}\\
A 1-body pure state has eigenvalues $(1,0,\dots,0)$, so it cannot be LME for $d\geq2$. The other case is identical by particle-hole duality. 
\item \textit{$d\geq2$ even and $N=2, N=d-2$}\\
For any state $\ket{\Psi}\in\wedge^2\mathbb{C}^d$, there are numbers $c_1\geq\dots\geq c_{\lfloor d/2\rfloor}\geq0$ and an orthonormal basis $\ket{u_1},\dots,\ket{u_d}$ such that [\onlinecite{yang1962},\onlinecite{youla1961}]
\begin{equation}
\label{YangYoula}
\ket\Psi=\sum_{j=1}^{\lfloor d/2\rfloor}c_{j}\ket{u_{2j-1}\wedge u_{2j}}.
\end{equation}
To obtain an LME state for $d$ even, we need $c_1=\dots=c_{d/2}=\sqrt{2/d}$. It is then the choice of basis that defines the LME state, but this can be changed with $K=SU(d)$ so $\dim(V^{N,d}_{\LME}/K)=0$. Particle-hole duality gives the same for $N=d-2$.
\item \textit{$d\geq2$ odd and $N=2, N=d-2$}\\
If $d$ is odd, the general form \eqref{YangYoula} rules out the existence of LME states. 
\item \textit{$d=6$ and $N=3$}\\
The following state is LME.
\begin{equation}
\frac{1}{\sqrt{2}}(\ket{u_1\wedge u_2\wedge u_3}+\ket{u_4\wedge u_5\wedge u_6}).
\end{equation}
\item \textit{$d=7$ and $N=3, N=4$}\\
The following state is LME.
\begin{equation}
\begin{aligned}
\frac{1}{\sqrt{7}}(\ket{u_1\wedge u_2\wedge u_3}&+\ket{u_1\wedge u_4\wedge u_5}+\ket{u_1\wedge u_6\wedge u_7}\\&+\ket{u_2\wedge u_4\wedge u_6}+\ket{u_2\wedge u_5\wedge u_7}+\ket{u_3\wedge u_4\wedge u_7}+\ket{u_3\wedge u_5\wedge u_6}).
\end{aligned}
\end{equation}
\item \textit{$d=8$ and $N=3, N=5$}\\
The following state is LME.
\begin{equation}
\begin{aligned}
\frac{1}{\sqrt{8}}(\ket{u_1\wedge u_2\wedge u_3}&+\ket{u_1\wedge u_4\wedge u_5}+\ket{u_1\wedge u_6\wedge u_7}+\ket{u_2\wedge u_4\wedge u_6}\\&+\ket{u_2\wedge u_5\wedge u_8}+\ket{u_3\wedge u_5\wedge u_7}+\ket{u_3\wedge u_6\wedge u_8}+\ket{u_4\wedge u_7\wedge u_8}).
\end{aligned}
\end{equation}
\item \textit{$d=8,N=4$, or $d\geq9$ and $3\leq N\leq d-3$}\\
We rely on [\onlinecite{Raamsdonk1}]. To comply with notation, set $V:=\wedge^N\mathbb{C}^d$. The groups $K:=\SU(d)$ and $G:=\SL(d)$ act on V symmetrically, that is $A\longmapsto A\otimes\dots\otimes A$. For the Lie algebras, this defines a representation $a\in \spl(d)\longmapsto a\otimes\mathds{1}\otimes\dots\otimes\mathds{1}+\dots+\mathds{1}\otimes\dots\otimes\mathds{1}\otimes a$. The moment map $\mu:\mathbb{P}(V)\to \spl(d)^*$ can then be written in terms of the 1-body reduced density matrix,
\begin{equation}
\mu(\ket{\Psi})(a)=\bra{\Psi}a\otimes\mathds{1}\otimes\dots\otimes\mathds{1}+\dots+\mathds{1}\otimes\dots\otimes\mathds{1}\otimes a\ket{\Psi}=\Tr[a\gamma^\Psi_1],
\end{equation}
by antisymmetry of $\ket\Psi$. The moment map maps to zero if and only if $\gamma_1^\Psi$ is proportional to the identity, but this happens if and only if $\ket\Psi$ is LME. Hence $V^{N,d}_{\LME}/K=\mu^{-1}(0)/K$.

We now simply apply the steps from [\onlinecite{Raamsdonk1}]. The recipe is as follows
\begin{enumerate}[label=\alph*)]
\item If $\rho:G\rightarrow \GL(V)$ is a representation of a complex reductive group $G$, and $V$ has a norm that is invariant under a maximal compact subgroup $K$ of $G$, the Kempf--Ness theorem [\onlinecite{kempf1979length}] applies and we have 
\begin{equation}
\mu^{-1}(0)/K\simeq \mathbb{P}(V)//G,
\end{equation}
which is the geometric invariant theory quotient of the projective space $\mathbb{P}(V)$. 

\item The dimension of this quotient $\mathbb{P}(V)//G$ is then derived in [\onlinecite{Raamsdonk1}] using two facts. The first is that, under the additional assumption that the representation $\rho$ is finite-dimensional, there exists a `generic'  stabilizer group $S$ [\onlinecite{richardson1972}] such that $\dim(\mathbb{P}(V)//G)=\dim(V)-\dim(G)+\dim(S)-1$. This $S$ is defined to be a closed subgroup of $G$ such that there exists an open dense subset $U\subset V$ with the property that for every $x\in U$, the stabilizer $G_x$ at $x$ is conjugate to $S$.

\item All that remains is to determine the dimension of $S$. This is done with work of \'Elashvili [\onlinecite{elashvili1972}], which says that assuming $G$ is semisimple and $\rho$ irreducible, we should check whether 
\begin{equation}
\label{criterion}
l(\left.\rho\right|_{H})\geq 1\ \text{for every (non-trivial) simple normal subgroup $H$ of $G$}.
\end{equation}
Given a faithful finite-dimensional representation $\rho:H\to \GL(V)$ of a simple complex linear algebraic group $H$, this index is defined as [\onlinecite{andreev1967}]
\begin{equation}
l(\rho):=\frac{\Tr[\rho^*(a)^2]}{\Tr[\ad(a)^2]},
\end{equation}
where $a\in \Lie(H)$, $\rho^*$ is the representation of $\Lie(H)$ associated with $\rho$ and $\ad$ is the adjoint representation of $\Lie(H)$. This is independent of the choice of $a$ as long as $\Tr[\ad(a)^2]\neq0$. If the criterion \eqref{criterion} holds, \'Elashvili [\onlinecite{elashvili1972}] provides us with the dimension $\dim(S)$, allowing for a calculation of $\dim(\mathbb{P}(V)//G)$.
\end{enumerate}

It is easy to check all the required assumptions hold and the recipe can be applied. We just need to verify \eqref{criterion} for $\SL(d)$. To calculate the index, take $a\in \spl(d)$ to be $a=\diag(\mu_1,\dots,\mu_d)$ with $\Tr[a]=\sum_i\mu_i=0$. As shown in Example 3.4 in [\onlinecite{andreev1967}], it is easy to calculate
\begin{equation}
\begin{aligned}
\Tr_{\spl(d)}[\ad(a)^2]&=\sum_{i\neq j}(\mu_i-\mu_j)^2=\sum_{i\neq j}\mu^2_i+\mu^2_j-2\mu_i\mu_j\\
&=2(d-1)\Tr[a^2]-2(\Tr[a]^2-\Tr[a^2])=2d\Tr[a^2].
\end{aligned}
\end{equation}
For
\begin{equation}
\label{somestep1}
\Tr[\rho^*(a)^2]=\Tr_{\wedge^N\mathbb{C}^d}\big[(a\otimes\mathds{1}\otimes\dots\otimes\mathds{1}+\dots+\mathds{1}\otimes\dots\otimes\mathds{1}\otimes a)^2\big],
\end{equation}
we use a basis of Slater determinants \eqref{Slater} built from the eigenvectors $\ket{u_1},\dots,\ket{u_d}$ of $a$. A single Slater determinant contributes terms of the form 
\begin{equation}
\bra{u_{i_1}\wedge\dots\wedge u_{i_N}}a^2\otimes\mathds{1}\otimes\dots\otimes\mathds{1}\ket{u_{i_1}\wedge\dots\wedge u_{i_N}}=\frac{1}{N}\sum_{1\leq k\leq N}\mu^2_{i_k},
\end{equation}
and similarly, 
\begin{equation}
\bra{u_{i_1}\wedge\dots\wedge u_{i_N}}a\otimes a\otimes \mathds{1}\otimes\dots\otimes\mathds{1}\ket{u_{i_1}\wedge\dots\wedge u_{i_N}}=\frac{1}{\tbinom{N}{2}}\sum_{1\leq k<k'\leq N}\mu_{i_k}\mu_{i_{k'}}.
\end{equation}
Noticing that the contribution $\mu^2_{i_k}$ is obtained from the $\tbinom{d-1}{N-1}$ Slaters that contain $\ket{u_{i_k}}$, and each contribution $\mu_{i_k}\mu_{i_{k'}}$ is obtained from the $\tbinom{d-2}{N-2}$ Slaters that contain both $\ket{u_{i_k}}$ and $\ket{u_{i_{k'}}}$, we find \eqref{somestep1} becomes
\begin{equation}
\begin{aligned}
\Tr[\rho^*(a)^2]&=\tbinom{d-1}{N-1}\sum_i\mu^2_i+2\tbinom{d-2}{N-2}\sum_{1\leq i<j\leq d}\mu_i\mu_j\\
&=\left[\tbinom{d-1}{N-1}-\tbinom{d-2}{N-2}\right]\Tr[a^2]+\tbinom{d-2}{N-2}\Tr[a]^2=\tbinom{d-2}{N-1}\Tr[a^2].
\end{aligned}
\end{equation}
Therefore, the index of the representation $A\otimes\dots\otimes A$ of $\SL(d)$ on $\wedge^N\mathbb{C}^d$ is 
\begin{equation}
\label{index}
l(\rho)=\frac{1}{2d}\binom{d-2}{N-1}.
\end{equation}
We check $l(\rho)=5/4$ for $d=8,N=4$. For $d\geq9$ and $3\leq N\leq d-3$, note 
\begin{equation}
\frac{1}{2d}\binom{d-2}{N-1}\geq \frac{1}{2d}\binom{d-2}{2}=\frac{d^2-5d+6}{4d}\geq \frac{7}{6}.
\end{equation}
where the first inequality is obvious from the properties of binomial coefficients, and the second can easily be derived by noting that the derivative in $d\geq9$ is positive.

Since \eqref{criterion} holds, [\onlinecite{elashvili1972}] says that the connected component $S^0$ of $S$ is trivial, but then $\dim(S)=\dim(S^0)=0$, and $\dim(\mathbb{P}(V)//G)=\dim(V)-\dim(G)-1=\tbinom{d}{N}-d^2$.
\end{proof}

\end{document}